\documentclass{easychair}
\newcommand{\I}{\emph{I}\xspace}
\newcommand{\You}{\emph{You}\xspace}
\newcommand{\My}{\emph{My}\xspace}

\newcommand{\Me}{\emph{Me}\xspace}
\newcommand{\Your}{\emph{Your}\xspace}

\newcommand{\DS}{\mathbf{DS}}
\newcommand{\dDS}{\mathbf{dDS}}

\newcommand{\At}{\mathsf{At}}
\renewcommand{\emptyset}{\varnothing}

\newcommand\ie{\hbox{\textit{i.e.}}}

\newcommand{\seq}{\Rightarrow}
\newcommand{\rs}[3]{#1;#2\seq #3}
\newcommand{\rc}{\mathcal{R}}

\newcommand{\M}{\mathbb{M}}

\newcommand{\A}{\mathsf{A}}
\newcommand{\R}{\mathsf{R}}
\newcommand{\V}{\mathsf{V}}
\newcommand{\g}{\mathsf{g}}
\newcommand{\ag}{\mathsf{a}}
\renewcommand{\b}{\mathsf{b}}
\renewcommand{\c}{\mathsf{c}}

\renewcommand{\qed}{\hfill$\blacksquare$}

\newtheorem{example}{Example}
\newtheorem{lemma}{Lemma}

\newtheorem{theorem}{Theorem}
\newtheorem{definition}{Definition}
\newtheorem{proposition}{Proposition}
\newtheorem{remark}{Remark}
\newtheorem{corollary}{Corollary}
\renewenvironment{proof}{\noindent\textit{Proof:}\quad}{\qed}

\newcommand{\pnlP}{\langle \bigdoublewedge+ \rangle }
\newcommand{\pnlN}{\langle\bigdoublewedge-\rangle}
\newcommand{\pnlPN}{\langle\bigdoublewedge\pm\rangle}
\newcommand{\pnlOP}{\langle\oplus\rangle}
\newcommand{\pnlON}{\langle\ominus\rangle}
\newcommand{\dplus}{\meddiamondplus}
\newcommand{\dminus}{\meddiamondminus}
\newcommand{\bplus}{\boxplus}
\newcommand{\bminus}{\boxminus}
\newcommand{\dplusminus}{\Diamond^{\pm}}

\DeclareMathOperator*{\bigdoublewedge}{\wedge\mkern-15mu\wedge}

\newcommand{\bbM}{\mathbb{M}}
\newcommand{\bfP}{\mathbf{P}}
\newcommand{\bfO}{\mathbf{O}}
\newcommand{\Rinv}{R_{\leftrightarrow}}

\newcommand{\ccpnlmodels}{\mathfrak{M}_{\textit{ccPNL}}}
\newcommand{\pnlmodels}{\mathfrak{M}_{\textit{PNL}}}
\newcommand{\namedmodels}{\mathfrak{M}_{\textit{N}}}

\newcommand{\PNL}{\textbf{PNL}}

\newcommand{\cc}{\textbf{cc}}
\newcommand{\dPNL}{\textbf{dPNL}}

\def\headline#1{\hbox to \hsize{\hrulefill\quad\lower.3em\hbox{#1}\quad\hrulefill}}

\newcommand\proofsystem[1]{\mbox{\slshape #1}\xspace}

\newcommand\aLKF {\hbox{\proofsystem{LKF}\kern-2pt$^a$}\xspace}

\newcommand\aLJF {\hbox{\proofsystem{LJF}\kern-2pt$^a$}\xspace}

\newcommand{\veenldotsveen}{\veen\kern-5pt\ldots\kern-2pt\veen}

\newcommand{\isemp}[1]{\ifthenelse{\isempty{#1}}{\cdot}{#1}}

\newcommand{\veen}{\vee^{\!-}}

\newcommand{\mkpos}[1]{\partial\kern -1pt_{\scriptscriptstyle +}\kern -1pt(#1)}
\newcommand{\mkneg}[1]{\partial\kern -1pt_{\scriptscriptstyle -}\kern -1pt(#1)}

\newcommand{\delayop}{\ensuremath{\partial}}

\newcommand{\delpfr}{\delayop_{\scriptscriptstyle +}^{\kern 1pt r}}\newcommand{\delpfl}{\delayop_{\scriptscriptstyle +}^{\kern 2pt l}}\newcommand{\delnfr}{\delayop_{\scriptscriptstyle -}^{\kern 1pt r}}\newcommand{\delnfl}{\delayop_{\scriptscriptstyle -}^{\kern 2pt l}}

\newcommand{\biasn}[1]{\delta^-\kern -1pt(#1)}
\newcommand{\biasp}[1]{\delta^+\kern -1pt(#1)}

\newcommand*\mdelim[3]{\mathopen{}\left#1#3\right#2\mathclose{}}

\makeatletter
\newcommand*{\cxs}{\@ifnextchar\i{\cxs@two}{\@ifnextchar\bgroup{\cxs@one}{}}}
\newcommand*{\cxs@one}[1]{\def\cxs@{#1}\mdelim{\lbrace}{\rbrace}{\ifx\cxs@\empty\mkern 3mu\else #1\fi}\cxs@one@decor }
\newcommand*{\cxs@two}[3]{\def\cxs@{#3}\mdelim{\lbrace\strut^{#2}}{\rbrace}{\ifx\cxs@\empty\mkern 3mu\else #3\fi}\cxs@one@decor }
\def\cxs@one@decor{\@ifnextchar\dots{\@firstoftwo{\dotsm\cxs@one@decor}}{\cxs}}

\def\cx@delete@right#1*{{#1}^{\star}\cx@continuation}
\def\cx@delete@always#1{{#1}^{\ast}\cx@continuation}

\def\cx@delete@star#1*{\@ifnextchar*{\cx@delete@right{#1}}{\cx@delete@always{#1}}}

\newcommand*{\@makecontextual}[2]{
	\newcommand*{#1}{\@ifnextchar*{\cx@delete@star{#2}}{#2\cx@continuation}}}
\newcommand*{\cx@continuation}[1][]{_{#1}\cxs}

\@makecontextual{\Ex}{}

\@makecontextual{\Cx}{\mathcal{C}}
\@makecontextual{\Dx}{\Delta}
\@makecontextual{\Lx}{\Lambda}
\@makecontextual{\Px}{\Pi}
\@makecontextual{\Lxpr}{\Lambda'}
\@makecontextual{\Rx}{\Gamma}
\@makecontextual{\RxP}{\Gamma^{P}}
\@makecontextual{\Rxs}{\Gamma^*}
\@makecontextual{\Rxpr}{\Gamma'}
\@makecontextual{\Dxb}{\Delta^{\rt{\bot}}}
\@makecontextual{\Dxs}{\Delta^*}
\@makecontextual{\Rxb}{\Gamma^{\rt{\bot}}}
\@makecontextual{\Lxb}{\Lambda^{\rt{\bot}}}
\@makecontextual{\Dxn}{\lf\Delta}
\@makecontextual{\Dxt}{\rt\Delta}
\@makecontextual{\Dxp}{\ct\Delta}
\@makecontextual{\Rxp}{\Pi^{+}}
\@makecontextual{\Rxm}{\Pi^{-}}
\@makecontextual{\Dxm}{\Delta^{-}}
\@makecontextual{\Ox}{\Omega}

\newcommand*{\BR}{\@ifnextchar\i{\br@two}{\@ifnextchar\bgroup{\br@one}{}}}
\newcommand*{\br@one}[1]{\def\br@{#1}\mdelim{\lbrack}{\rbrack}{\ifx\br@\empty\mkern 3mu\else #1\fi}}
\newcommand*{\br@two}[3]{\def\br@{#3}\mdelim{\lbrack\strut^{#2}}{\rbrack}{\ifx\br@\empty\mkern 3mu\else #3\fi}}

\newcommand*{\@makeoperator}[2]{
	\newcommand*{#1}{\mathrm{#2}\mdelim{(}{)}
	}
}
\@makeoperator{\fm}{et}

\makeatother

 \usepackage[utf8]{inputenc}
\usepackage{pmboxdraw}
\usepackage{fancyvrb}
\usepackage{xcolor}
\usepackage{etoolbox}
\usepackage{amsmath,amssymb,trimclip,adjustbox}
\usepackage{stackengine,amssymb,graphicx}
\usepackage{color}
\usepackage{colortbl}
\usepackage{xspace}

\usepackage{modalops}
\usepackage{xifthen}

\usepackage{multicol}
\usepackage{longtable}

\usepackage{ebproof}	

\usepackage{proof}

\usepackage{comment} \usepackage{tabularx}

\usepackage{cancel}
\usepackage[colorinlistoftodos,prependcaption,textsize=tiny]{todonotes}

\usepackage{stackengine}

\usepackage{cleveref}
 \usepackage{doc}
\usepackage{tikz}
\usetikzlibrary{positioning, quotes}
\DeclareUnicodeCharacter{2714}{\textcolor{teal}{$\checkmark$}}
\DeclareUnicodeCharacter{274C}{\textcolor{red}{$\times$}}

\DeclareUnicodeCharacter{25C6}{$\dplus$}
\DeclareUnicodeCharacter{25C7}{$\dminus$}
\DeclareUnicodeCharacter{2228}{$\vee$}
\DeclareUnicodeCharacter{2227}{$\wedge$}

\title{Reasoning About Group Polarization:\\
From Semantic Games to Sequent Systems}

\author{
Robert Freiman\inst{1}\thanks{Research supported by FWF project P 18563.}
\and
Carlos Olarte\inst{2}\thanks{The work of Olarte has been partially supported by the SGR project PROMUEVA (BPIN
2021000100160) under the supervision of Minciencias (Ministerio de Ciencia Tecnolog\'ia e Innovaci\'on, Colombia). Olarte acknowledge also support from the NATO
Science for Peace
and Security Programme through grant number G6133 (project SymSafe). }
\and
   Elaine Pimentel\inst{3}\thanks{Pimentel has received funding from the European Union's Horizon 2020 research and innovation programme under the Marie Sk\l odowska-Curie grant agreement Number 101007627.}
\and
    Christian G.\ Ferm\"uller\inst{1}}

\institute{
  TU-Wien, Austria\\
  \email{\{robert,chrisf\}@logic.at}
\and
     LIPN, CNRS UMR 7030, Universit\'{e} Sorbonne Paris Nord, France \\
   \email{olarte@lipn.univ-paris13.fr}\\
\and
   Computer Science Department UCL, UK\\
   \email{e.pimentel@ucl.ac.uk}
 }

\authorrunning{Freiman, Olarte, Pimentel and Ferm\"uller}

\titlerunning{Reasoning About Group Polarization}

\begin{document}

\maketitle 

\begin{abstract}
    Group polarization, the phenomenon where individuals  become more
extreme after interacting, has been gaining attention, especially with the rise of
social media  shaping people's opinions. Recent
interest has emerged in formal reasoning about group polarization using logical
systems. In this work we consider the modal logic PNL that captures the notion
of agents agreeing or disagreeing on a given topic. 
Our contribution involves enhancing PNL with advanced formal reasoning techniques, instead of relying on axiomatic systems for analyzing group polarization. To achieve this, we introduce a semantic game tailored for (hybrid) extensions of PNL. This game fosters dynamic reasoning about concrete network models, aligning with our goal of strengthening PNL's effectiveness in studying group polarization.
We show how this
semantic game leads to a provability game by systemically exploring the truth
in all models. This leads to the first  cut-free  sequent systems
for some variants of PNL. Using polarization of formulas, the proposed calculi
can be modularly adapted to consider different frame properties of the
underlying model. 

 \end{abstract}

\section{Introduction}\label{sec:intro} Group polarization -- where the
opinions or beliefs of individuals within a group become more extreme or
polarized after interacting with each other -- is rapidly gaining attraction,
especially with the advent of social media platforms that have played a key role
in the polarization of social, political, and democratic processes. This
phenomenon is mainly studied in psychology
\cite{myers1976group,isenberg1986group} and political philosophy
\cite{sunstein1999law,sunstein2007group}. More recently, logicians have taken
up the challenge of formal reasoning about social networks and changes in 
agents's beliefs. Take for instance 
the Facebook logic  \cite{DBLP:conf/tark/SeligmanLG13} (an
epistemic logic endowed with a symmetric ``friendship'' relation); the Tweeting logic
\cite{DBLP:conf/lori/XiongASZ17} (formalizing announcements in a network); 
the logic for reasoning about social belief and change propagation
\cite{DBLP:journals/synthese/LiuSG14}, etc. 

We  focus on the modal logic \PNL\ 
\cite{DBLP:journals/jolli/XiongA20}, which refers to Kripke frames with two
types of disjoint and symmetric reachability relations. The individuals in a
social network are identified with worlds of the frame, and they are related
either as ``friends'' (positive) or as ``enemies'' (negative), but not both at
the same time. These relationships can be understood in different ways: Instead
of genuine friendship or enduring enmity, they may simply signify agreement or
disagreement on a particular issue. Unlike the aforementioned logics, \PNL~was
designed to reason about positive and negative relations among agents, a key
aspect for defining and measuring polarization \cite{Bramson2017}. In fact,
polarization can actually be studied in a neutral, network theory framework.
Under the framework of balance theory, it is possible to investigate the essential conditions required for network stability in a fully polarized social network.
For instance, the \emph{local balance} condition that prohibits
triangles of nodes with two positive and one negative connection can be
associated with the formation of clusters of pairwise positively connected
nodes that are negatively connected to all nodes outside the cluster (see
\Cref{ex:balance}). 

We take inspiration from a work by Pedersen, Smets, and {\AA}gnotes
\cite{DBLP:journals/logcom/PedersenSA21}, where \PNL~is extended in
various ways to axiomatically characterize modally undefinable frame
properties, including the disjointness of the two relations and collective
connectedness. The main challenge is the axiomatization of the balance
property, which requires extensions of \PNL~with nominals,
 dynamic and hybrid operators. 
 
Our approach to logical reasoning about group
polarization is also based on \PNL~but focuses on a different aspect
of formal reasoning about the corresponding models {via games and proof systems. Games are a powerful tool to bridge the gap between intended and formal semantics, often offering a conceptually more natural approach to logic than the common paradigm of model-theoretic semantics.  
In semantic games~\cite{Hintikka1973-HINLLA-2}, every instance of the game is played over a formula $\phi$ and a model $\M$ by two players, usually called \I (or \Me) and \You. At each point in the game, one of the players acts as the proponent ($\mathbf P$), while the other acts as the opponent ($\mathbf O$) of the current formula. The set of actions at each stage is dictated by the main connective of the current formula. 
In contrast to semantic games, provability games~\cite{Lorenzen1978-LORDLJ-2} do not refer to truth in a particular model but to {\em logical validity}. The game is also played by two players, \Me and \You, and consists of attacking assertions of formulas made by the other player and defending against these attacks. In this work, we will introduce both a semantic game and a provability game for (hybrid) extensions of \PNL.}

{We start by proposing a 
semantic game that characterizes
the truth in a given network model. This provides an
alternative to the standard definition of an evaluation function which supports
a dynamic form of reasoning about concrete network models (\Cref{sec:pnl}).}
We move on by arguing that effective formal reasoning with the relevant logics requires more
than (just) Hilbert-style axiom systems. Rather, the automated search for
proofs calls for Gentzen-style systems that respect ({a restricted form of}) the
subformula property. In proof-theoretic terms, we are looking for a cut-free
sequent system. Hence, our next step is to turn the semantic game over single
models into a provability game (\Cref{sec:dis-game}), characterizing logical validity. To this end,
we define disjunctive states for a game that is not restricted to a single
model, but systematically explores the truth in all models. This method leads
to {\em the first} Gentzen-style systems  for variants of \PNL~(\Cref{sec:proofs}), which
modularly adapts to different frame properties by faithfully capturing the rules 
for \emph{elementary} games.

Models of social learning and opinion dynamics aim to understand the role of certain social factors in the acceptance/rejection of opinions. They can be useful to explain alternative scenarios, such as consensus or polarization. In this context, the positive and negative relationships are not permanent. Instead, they can vary over time when
\emph{enemies} reconcile, new \emph{friendships}/agreements emerge, or some actors begin to disagree with others. In \Cref{sec:extensions}, we show how the {\em global adding} and {\em local link change}  modalities  of
\cite{DBLP:journals/logcom/PedersenSA21} (inspired by sabotage modal logic
\cite{DBLP:journals/igpl/ArecesFH15,DBLP:journals/logcom/AucherBG18, DBLP:journals/logcom/BenthemLSY23})
can be defined in our framework.
{As a plus, we present in \cite{tool} a 
prototypical implementation of the proposed games 
using rewriting logic and Maude \cite{DBLP:journals/jlp/Meseguer12,DBLP:journals/jlap/DuranEEMMRT20}.}

\section{A Game Semantics for PNL}
\label{sec:pnl}

In this section, we revisit the positive and negative relations logic
\cite{DBLP:journals/jolli/XiongA20,DBLP:journals/logcom/PedersenSA21} with
nominals (\PNL) and its standard Kripke semantics, proposing a novel semantic
game  for \PNL~that we prove to be adequate. Paving the way for the provability
game introduced in Section \ref{sec:dis-game}, we also propose an alternative
presentation of \PNL~that internalizes the nominals.

\subsection{Kripke semantics for \PNL}\label{sec:pnl-k}
Let $\A=\{\ag,\b,\ldots\}$ be a non-empty set of agents,
$\At=\{p,q,\ldots\}$ be a countable set of propositional variables, and $N=\{i,j,\ldots\}$ be a countable set of \emph{nominals}. The language of \PNL~is generated by the following grammar:
$$\phi ::= p  \mid R^+(i,j)\mid R^-(i,j) \mid \neg \phi \mid \phi \wedge \phi \mid \phi \vee \phi \mid \dplus \phi \mid \dminus \phi\mid [A]\phi$$
where $p\in \At$, and $i,j\in N$. 
The propositional connectives $\top$, $\bot$, $\to$, and the (dual)
 modalities $\bplus$ and $\bminus$ can be obtained in the usual way.

 Formulas of the form $p$, $R^+(i,j)$, or $R^-(i,j)$ are called \emph{elementary}.
The proposition $R^+(i,j)$ states that $i$ is a \emph{friend} of (or, more generally, \emph{agrees} with)
 $j$, and proposition $R^-(i,j)$ states that agent $i$ is an \emph{enemy} of (or \emph{disagrees} with) $j$. 
The formula $\dplus \phi$ (resp. $\dminus \phi$) states that $\phi$ holds for  a friend (resp.\ an enemy). The global
modality $[A]\phi$ states that $\phi$ holds for all the agents. 
We use $R^\pm$ to denote either $R^+$ or $R^-$, and 
 $\dplusminus$ to denote either $\dplus$ or $\dminus$.

A model $\mathbb{M}$ is a tuple $\langle \A,\R^+,\R^-,\V,\g\rangle$ where $\A$
is a set (of agents), $\g:N\rightarrow \A$ is called \emph{denotation
function}, $\R^+,\R^-\subseteq \A\times \A$, and $\V:\At\rightarrow
\mathcal{P}(\A)$. A model is a \PNL-model if 
$\R^+$ is reflexive,  and 
$\R^+$ and $\R^-$ are both symmetric and 
non-overlapping, i.e.,  for all $\ag,\b\in \A$, $(\ag,\b)\notin \R^+$ or $(\ag,\b)\notin \R^-$. 
We say that a model $\M$ is \emph{collectively connected}, or a \cc-\PNL-model,
if,  additionally,  for all $\ag,\b\in \A$,   $(\ag,\b)\in \R^+$ or $(\ag,\b)\in \R^-$.
The class of all \PNL~models (\cc-\PNL-models) is denoted by  $\pnlmodels$ ($\ccpnlmodels$).

\begin{figure}
$\small
\qquad\begin{array}{lll l lll}
    \mathbb{M},\ag \Vdash p &\text{ iff } \ag\in \V(p) &\quad& 
    \mathbb{M},\ag \Vdash \neg \phi &\text{ iff } \mathbb{M},\ag \not \Vdash \phi\\
    \mathbb{M},\ag \Vdash R^+(i,j) &\text{ iff } (\g(i),\g(j))\in \R^+ & \quad & 
    \mathbb{M},\ag \Vdash R^-(i,j) &\text{ iff } (\g(i),\g(j))\in \R^-\\
    \mathbb{M},\ag \Vdash \phi \wedge \psi &\text{ iff } \mathbb{M},\ag \Vdash \phi \text{ and } \mathbb{M},\ag \Vdash \psi &\quad&
    \mathbb{M},\ag \Vdash \phi \vee \psi &\text{ iff } \mathbb{M},\ag \Vdash \phi \text{ or } \mathbb{M},\ag \Vdash \psi\\
\mathbb{M},\ag \Vdash \dplus\phi & \multicolumn{6}{l}{\text{ iff there is } \b\in \A \text{ such that } (\ag,\b)\in \R^+ \text{ and } \mathbb{M},\b \Vdash \phi}\\
    \mathbb{M},\ag \Vdash \dminus\phi &\multicolumn{6}{l}{\text{ iff there is } \b\in \A \text{ such that } (\ag,\b)\in \R^- \text{ and } \mathbb{M},\b \Vdash \phi}\\
    \mathbb{M},\ag \Vdash [A]\phi &\multicolumn{6}{l}{\text{ iff } \mathbb{M},\b \Vdash \phi \text{ for all } \b\in \A}
\end{array}
$
\caption{Kripke semantics for \PNL~\cite{DBLP:journals/logcom/PedersenSA21}\label{fig:ksem}.}
\end{figure}

The Kripke semantics of \PNL~is in \Cref{fig:ksem}.
A formula $\phi$ is true over $\mathbb{M}$, written $\mathbb{M}\Vdash \phi$ iff
$\mathbb{M},\ag\Vdash \phi$, for all agent $\ag\in\A$. For a set of formulas
$\Phi$, we write $\mathbb{M}\models \Phi$ iff $\mathbb{M}\Vdash \Phi$ for all $\phi
\in \Phi$. A formula $\phi$ is ((\cc-)\PNL-) valid iff $\mathbb{M}\Vdash\phi$ for every
((\cc-)\PNL)-model $\mathbb{M}$. For a class of models $\mathfrak{M}$, we write
$\Phi\models_{\mathfrak{M}} \phi$ iff $\mathbb{M}\Vdash \phi$ for every model
$\mathbb{M}\in \mathfrak{M}$ with $\mathbb{M}\models \Phi$.

A model is called \emph{named} if its denotation function is surjective, i.e., every
agent has a name. Let $\namedmodels$ be the class of named models. The following result is immediate. 

\begin{lemma}
Let $\Phi\cup\{\phi\}$ be a finite set of formulas. Then
$\Phi\Vdash_{\pnlmodels}\phi \iff \Phi\Vdash_{\pnlmodels\cap \namedmodels}\phi$ and $\Phi\Vdash_{\ccpnlmodels}\phi \iff \Phi\Vdash_{\ccpnlmodels\cap \namedmodels}\phi$.
\end{lemma}

Using this lemma, we give an alternative presentation of the semantics 
where we explicitly test, using elementary formulas,  the existence of 
$R^{\pm}$-successors.
Let $\ag=\g(i)$. Then, we define: \\

$
\begin{array}{lll}
    \mathbb{M},\ag \Vdash \dplusminus\phi &\text{ iff there is } j\in N \text{ such that } \mathbb{M},\g(j)\Vdash R^{\pm}(i,j) \text{ and }  \mathbb{M},\g(j)\Vdash \phi\\
\mathbb{M},\ag \Vdash [A]\phi &\text{ iff } \mathbb{M},\g(j)\Vdash \phi, \text{ for all } j\in N, 
\end{array}
$

\begin{remark}\label{rem:surjg}
At a later point, we need the following observation: if $\g$ is surjective, even
if restricted to $N'\subseteq N$, then  $N$ in the above truth conditions
can equivalently be replaced by $N'$. \end{remark}

\subsection{Semantic Game}\label{sec:game-semantics}
The proposed \emph{semantic game} is played over a \PNL-model $\M=(\A,\R^+,\R^-,\V,\g)$ by two
players, \Me and \You, who argue about the truth of a formula $\phi$ at an
agent $\ag$. At each stage of the game, one player acts as \emph{proponent}, while the
other acts as \emph{opponent} of the claim that  $\phi$ is true at 
$\ag$. We represent the situation where \I am the proponent (and \You are the
opponent) by the \emph{game state} $\mathbf{P}, \ag:\phi$, and the situation
where \I am the opponent (and \You are the proponent) by $\mathbf{O},
\ag:\phi$. We call a game state \emph{elementary} if its involved formula is
elementary. For a game state $g$, we denote the game starting at $g$ over the
model $\M$ by $\mathbf{G}_\M(g)$.

The game proceeds by reducing the involved formula $\phi$ to an elementary formula. The following rules of the game describe the possible choices of the players depending on the current game state,
when playing a game over the model $\M$.

\begin{definition}
    Let $\M$ be a \PNL-model. The semantic game is defined by the rules in \Cref{fig:game-rules}. 
\end{definition}

\begin{figure}
~~~\resizebox{.92\textwidth}{!}{
\begin{minipage}[t]{\textwidth}
    \hrulefill
\begin{description}
\item[$(\mathbf{P}_\wedge)$] At $\mathbf{P}, \ag: \phi_1\wedge \phi_2$, \You
    choose between  $\mathbf{P},\ag:\phi_1$ and 
    $\mathbf{P},\ag:\phi_2$ to continue the game. 
\item[$(\mathbf{O}_\wedge)$] \vspace{-2mm}At $\mathbf{O}, \ag:
    \phi_1\wedge \phi_2$, \I choose between $\mathbf{O},\ag:\phi_1$
    and $\mathbf{O},\ag:\phi_2$ to continue the game. 

\item[$(\mathbf{P}_\vee)$] At $\mathbf{P}, \ag: \phi_1\vee \phi_2$, \I choose between $\mathbf{P},\ag:\phi_1$ and $\mathbf{P},\ag:\phi_2$ to continue the game
\item[$(\mathbf{O}_\vee)$] \vspace{-2mm}At $\mathbf{O}, \ag: \phi_1\vee \phi_2$, \You choose between $\mathbf{O},\ag:\phi_1$ and $\mathbf{O},\ag:\phi_2$ to continue the game.

\item[$(\mathbf{P}_\neg)$] At $\mathbf{P}, \ag: \neg \phi$, the game continues with $\mathbf{O}, \ag: \phi$.
\item[$(\mathbf{O}_\neg)$] \vspace{-2mm}At $\mathbf{O}, \ag: \neg \phi$, the game continues with $\mathbf{P},\ag: \phi$.

\item[$(\mathbf{P}_{\Diamond^\pm})$] At $\mathbf{P}, \ag: \Diamond^\pm \phi$, \You win if there are no $\R^\pm$-successors of $\ag$. Otherwise, \I choose an $\R^\pm$-successor $\b$ and the game continues with $\mathbf{P},\b:\phi$.
\item[$(\mathbf{O}_{\Diamond^\pm})$] \vspace{-2mm} At $\mathbf{O},\ag: \Diamond^\pm \phi$, \I win if there are no $\R^\pm$-successors of $\ag$. Otherwise, \You choose an $\R^\pm$-successor $\b$ and the game continues with $\mathbf{O},\b:\phi$.

\item[$(\mathbf{P}_{[A]})$] At $\mathbf{P},\ag :  [A]\phi$, \You choose an agent $\b$ and the game continues with $\mathbf{P},\b:\phi$.
\item[$(\mathbf{O}_{[A]})$] \vspace{-2mm}At $\mathbf{O},\ag :  [A]\phi$, \I choose an agent $\b$ and the game continues with $\mathbf{O},\b:\phi$.

\item[$(\mathbf{P}_{el})$] Let $\phi_{e}$ be an elementary formula. \I win and \You lose at $\mathbf{P},\ag:\phi_{e}$ iff $~\mathbb{M},\ag \models \phi_{e}$. Otherwise, \You win and \I lose.
\item[$(\mathbf{O}_{el})$] \vspace{-2mm}At $\mathbf{O},\ag : \phi_{e}$, \I win and \You lose iff $\mathbb{M},\ag\not \models \phi_{e}$. Otherwise, \You win and \I lose.
\end{description}
    \hrulefill
\end{minipage}
}
\caption{Semantic game given a \PNL-model $\M$.\label{fig:game-rules}}
\end{figure}

In general, every two-person, zero-sum, win-lose game is usually represented by
a \emph{game tree}, i.e., a labeled tree whose nodes are game states. In our
case, the root of the game tree representing the game
$\mathbf{G}_\mathbb{M}(g)$ is $g$. The children of each node in the game tree
are exactly the possible choices of the corresponding player. For instance, if
$h=\mathbf{P}, \ag: \phi_1\wedge \phi_2$ appears in the game tree, then its
children are $\mathbf{P},\ag:\phi_1$ and $\mathbf{P},\ag:\phi_2$. Each node in
the tree is labeled either ``I'', or ``Y'', depending on which player is to
move in the corresponding game state, and we label 
the nodes  $\mathbf{P}, \ag: \neg \phi$ and $\mathbf{O}, \ag:
\neg \phi$ with ``I'' (even though there is no choice involved in these game
states). For instance, the node corresponding to the game state $h$ above is
``Y'', since it is \Your choice in $\mathbf{P}:\phi_1\wedge \phi_2$. The leaves
of the tree receive the label of the winning player. 
 A \emph{run} of the game is a maximal path through the game tree.

 \begin{example}\label{ex:balance}
     Let 
     $ \mathbf{(4B)} = 
        ((\dplus\dplus p \vee \dminus\dminus p)\to \dplus p) \wedge
        ((\dplus \dminus p \vee \dminus\dplus p)\to \dminus p)
    $. This formulas 
    specifies \emph{local balance} \cite{DBLP:journals/logcom/PedersenSA21}
    and   captures the idea that 
``the enemy of my enemy is my friend'',  ``the friend
of my enemy is my enemy'',  and ``the friend of my friend is my friend''.
 A collectively connected network where $[A]4B$ holds is a polarized network, 
where agents can be divided into two opposing groups \cite{balance}.
Notions such a weak-balance \cite{weak-balance} can be also formalized in \PNL~\cite{DBLP:journals/logcom/PedersenSA21}.
    $\I$ have a winning strategy for the game $\bfP,\ag : 4B$
    on $\bbM_1$ 
    while $\You$ have a winning strategy for the same game on $\bbM_2$ where (omitting self-loops for $\R^+$):

    \begin{tabularx}{.6\textwidth}{ X X X X X }
        $\bbM_1=$ & \parbox[c]{\hsize}{\includegraphics[scale=0.85]{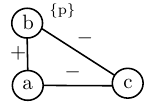}} & \qquad &
        $\bbM_2=$ & \parbox[c]{\hsize}{\includegraphics[scale=0.85]{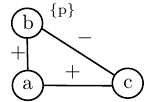}} 
    \end{tabularx}

    For $\bbM_1$, in the first conjunct, \I pick ($\bfP_\vee$)
    $\dplus p$ and then $\b$ in ($\bfP_{\dplus}$); for the second conjunct,
    \I pick the first disjunction in $F=(\dplus \dminus p \vee \dminus\dplus p)\to \dminus p)$
    where, in any of \Your choices ($\bfP_{\neg}$ followed by $\bfO_{\vee}$ and 
    $\bfO_{\Diamond^{\pm}}$), \I  win all the elementary states. 
    For $\bbM_2$, \I do not have a winning strategy 
    for the  second conjunct: \I can neither win $\dminus p$ (no $\R^{-}$ successor), 
    nor the first disjunct in  $F$ above since, after $\bfP_{\neg}$, 
    \You choose  ($\bfO_{\vee}$) 
    $\dplus\dminus p$ and select $\c$ and then $\b$ ($\bfO_{\Diamond^{\pm}}$)
    where $p$ holds and \You win. See the complete game in our tool \cite{tool}. 
\end{example}

The following proposition follows from the fact that every rule
in the game decomposes the involved formula in subformulas. 
\begin{proposition}\label{prop:finiteheight}
For all models $\M$ and game states $g$, the game tree of $\mathbf{G}_\M(g)$ is of finite height.
\end{proposition}

\noindent An immediate consequence of this fact is that every semantic game $\mathbf{G}_\M(g)$ is \emph{determined}, i.e., exactly one of the two players has a winning strategy.

\subsection{Strategies and adequacy}\label{sec:adequacy}
Now we are ready to define winning strategies and prove the main result of this section:
the adequacy of the proposed game semantics with respect to the Kripke semantics for \PNL.

\begin{definition}
    A \emph{strategy} for \Me in the game $\mathbf{G}_\M(g)$ is a subtree $\sigma$ of the associated game tree such that: 
 \textbf{(1)} $g\in \sigma$,
 \textbf{(2)} if $h\in \sigma$ is a node labeled ``Y'', then all children of $h$ are in $\sigma$,
 \textbf{(3)} if $h\in \sigma$ is a node labeled ``I'', then exactly one child of $h$ is in $\sigma$.
The strategy $\sigma$ is called \emph{winning} if all leaves in the tree $\sigma$ are labeled ``I''. (Winning) strategies for \You are defined dually.
\end{definition}

Note that every combination of strategies from \Me and \You defines a unique run of the game. Alternatively, we can define a strategy $\sigma$ for \Me to be winning iff the run resulting from $\sigma$ and a strategy for \You ends in a winning game state for \Me.

We can now show the adequacy of the semantic game, \ie, that the existence of winning strategies for Me in a game and truth in a model coincide (proof in Appendix~\ref{app:proofs}).
\begin{theorem}\label{th:adequacy}
Let $\M$ be a \PNL-model, $\ag$ an agent, and $\phi$ a formula.

\noindent\textbf{(1)} \I have a winning strategy for $\mathbf{G}_\M(\mathbf{P}, \ag: \phi)$ iff $\M,\ag \models \phi$. 

\noindent\textbf{(2)} \You have a winning strategy for $\mathbf{G}_{\M}(\mathbf{P}, \ag: \phi)$ iff $\M,\ag\not \models \phi$.
\end{theorem}

\paragraph{Internalizing nominals.}\label{sec:intern}
Remember that in a named model $\M$, every agent $\ag$ has a name $i$, i.e. there exists $i\in N$ s.t. $\g(i)=\ag$. Therefore, for $\mathbf{Q} \in \{\mathbf{P},\mathbf{O}\}$, it is unambiguous if we write $\mathbf{Q}, i:\phi$ for the game state $\mathbf{Q},\ag : \phi$. The fact that $\mathbb{M} \models R^{\pm}(i,j)$ iff $(\g(i),\g(j))\in \R^{\pm}$ gives us the following equivalent formulations of the rules for $\dplus$, $\dminus$ and $[A]$\footnote{The outcome of the game state $\mathbf{Q},k:R^{\pm}(i,j)$ is independent of $k$ (it only depends on the underlying model $\M$). Hence, we write 
        $\mathbf{Q},\_:R^\pm(i,j)$ instead of $\mathbf{Q},k:R^\pm(i,j)$.\label{foot:R}}:\\

\vspace{-0.2cm}
\resizebox{.96\textwidth}{!}{
\noindent\begin{minipage}[t]{\textwidth}
\begin{description}
    \item[$(\mathbf{P}_{\Diamond^\pm})$] At $\mathbf{P}, i: \Diamond^\pm \phi$, \I choose a nominal $j$, and \You decide whether the game ends in the state $\mathbf{P},\_:R^\pm(i,j)$ or continues with $\mathbf{P},j: \phi$.\item[$(\mathbf{O}_{\Diamond^\pm})$] \vspace{-2mm} At $\mathbf{O},i: \Diamond^\pm \phi$, \You choose  $j$, and \I choose between $\mathbf{O},\_:R^\pm(i,j)$ and $\mathbf{O},j:\phi$.
\item[$(\mathbf{P}_{[A]})$] At $\mathbf{P},i: [A]\phi$, \You choose a nominal $j$ and the game continues with $\mathbf{P},j:\phi$.
\item[$(\mathbf{O}_{[A]})$] \vspace{-2mm} At $\mathbf{O},i: [A]\phi$, \I choose a nominal $j$, and the game continues with $\mathbf{O},j: \phi$.
\end{description}
\end{minipage}
}\\

\noindent Following Remark~\ref{rem:surjg}, we can restrict branching over a subset $N'$
 of nominals if $g$'s restriction to  $N'$ is surjective. For future reference,
 let us formulate this observation precisely. We denote the game over
 $\mathbb{M}$ starting at the game state $g$, where branching is over
 $N'\subseteq N$ by $\mathbf{G}^{N'}_\M(g)$. We call a model $N'$-named if
 $\g:N'\rightarrow \A$ is surjective. We say that two games $\mathbf{G}_1$ and
 $\mathbf{G}_2$ are \emph{strategically equivalent}, notation $\mathbf{G}_1\cong
 \mathbf{G}_2$, iff \I have a winning strategy in both games or in none of the
 two. We have the following:

\begin{proposition}
    \label{prop:surjg}
Let $\M$ be $N'$-named, let $g=\mathbf{Q},\ag:\phi$ for some $\mathbf{Q}\in
\{\mathbf{P},\mathbf{O}\}$, some agent $\ag$ and formula $\phi$, and let
$\ag=\g(i)$. Then $\mathbf{G}_\M(\mathbf{Q},\ag:\phi)\cong
\mathbf{G}^{N'}_\M(\mathbf{Q},i:\phi)$. 
\end{proposition}

\section{The Disjunctive Game}\label{sec:dis-game}

In this section, we lift the semantic game to two different \emph{disjunctive
games}, one for \PNL-validity and one for \cc-\PNL-validity, which differ only in
their respective winning conditions. We prove that the obtained games are
adequate and can thus be regarded as \emph{provability games} for their
corresponding logics. The main crucial fact is that the rules of the semantic
game are independent of the underlying model, except for elementary game
states. Using this fact, and
identifying certain conditions on winning strategies, this disjunctive game
will be the foundation for the sequent calculi proposed in Section \ref{sec:proofs}. 

\subsection{Playing on all models}\label{sec:dgame}
The disjunctive game
$\mathbf{DG}(\mathbf{P},i:\phi)$ can be thought of as \Me and \You playing all
semantic games $\mathbf{G}(\mathbf{P},i:\phi)$ over all \PNL-models $\M$
simultaneously. We point out that  the rules of the semantic game do not depend on the structure
of $\M$ but merely on $\phi$. Truth degrees are only needed at the atomic level
to determine who wins the particular run of the game. This allows us to require
players to play ``blindly'', i.e., without explicitly referencing  a model $\M$.
Clearly, if \I have a winning strategy in such a game, then \I can win in
$\mathbf{G}_\M(\mathbf{P},i:\phi)$, for every $\M$, making this strategy an
adequate witness of logical validity. 

Defining such a validity game is not straightforward, since the simplest case of disjunction
is already problematic.  Let us consider 
the game $\mathbf{P},i: p \vee \neg p$.
Clearly, \I have a winning strategy in the semantic game over every model.
However, there is no uniform way of making a good choice in the first turn: No
matter whether \I choose $\mathbf{P},i: p$ or $\mathbf{P},i: \neg p$, there are
still models where \You win the game eventually. To compensate for this, we
allow \emph{Myself} to create ``backup copies'' and \emph{duplicate} game
states. Formally, disjunctive game states are finite multisets of the game states
defined in Section \ref{sec:game-semantics}. We prefer to write $g_1 \bigvee ... \bigvee g_n$ for
the disjunctive game state $\{g_1,...,g_n\}$, but keep the convenient notation
$g\in D$ if $g$ is in the  multiset $D$. We write $D_1 \bigvee D_2$ for the
multiset sum $D_1+D_2$ and $D\bigvee g$ for $D+\{g\}$. A disjunctive state is
called \emph{elementary} if all its game states are elementary.
We use $\mathbf{DG}(D)$ to denote the disjunctive game starting at $D$,
and we define the \cc-disjunctive game $\mathbf{DG}^\cc(D)$ which is played
over all $\cc$-$\PNL$-models. It differs from $\mathbf{DG}(D)$ only in its
winning conditions (see below).

In our running example, \I duplicate the game state in the first round and the
game continues with the \emph{disjunctive state} $\mathbf{P},i: p \vee \neg
p\bigvee \mathbf{P},i: p \vee \neg p$. Now \I move to $\mathbf{P},i: p$ in the
first subgame and to $\mathbf{P},i:\neg p$ in the second. After a role switch
in the second subgame, the final state is $\mathbf{P},i: p \bigvee
\mathbf{O},i: p$, where \I win regardless of the underlying model.

The following \emph{winning condition} reflects the fact that \My strategy for
the disjunctive game $\mathbf{DG}(\mathbf{P},i:p \vee \neg p)$ was successful.
\begin{definition}\label{def:win}
Let $D^{el}$ denote the disjunctive state consisting of the elementary game
states of $D$. \I win and \You lose at $D$ if for every \PNL-model there is
a game state in $D^{el}$ where \I win the corresponding semantic game. In the \cc-disjunctive game, \I win and \You lose if for every \cc-\PNL-model there is
a game state in $D^{el}$ where \I win the corresponding semantic game.
\end{definition}

In the disjunctive game, \I additionally take the
role of a \emph{scheduler}, deciding which game  is to be played next. We
signal the chosen game state by underlining it as in $\underline{g}$.

\begin{figure}
\resizebox{0.93\textwidth}{!}{
\begin{minipage}[t]{\textwidth}
\hrulefill
\begin{description}
\item[(Dupl)] If no state in $D$ is underlined,
\I can choose a non-elementary $g\in D$ and the game continues with $D\bigvee g$.
\item[(Sched)] If no state in $D=D'\bigvee g$ is underlined,
and $g$ is non-elementary, 
    \I can choose to continue the game  with $D'\bigvee \underline{g}$.
\item[(Move)] If $D=D'\bigvee \underline{g}$ then the player who is to move in
    the semantic game $\mathbf{G}(g)$ at $g$ makes a legal move to the game
    state $g'$ and the game continues with $D' \bigvee g'$.\footnote{For example, if $g$
    is $\mathbf{P},i:\phi_1 \wedge \phi_2$, then \You choose a $k\in \{1,2\}$
    and the game continues with $D'\bigvee \mathbf{P},i:\phi_k$.}
\item[(End)]
    The game ends if there are no non-elementary game states left in $D$, or if
    no game state is underlined and \I win according to
    Definition~\ref{def:win}. Otherwise, \I must move according to
    \textbf{(Dupl)} or \textbf{(Sched)}.
\end{description}
\hrulefill
\end{minipage}
}
\caption{Rules for the disjunctive game\label{fig:dg-rules}.}
\end{figure}

\begin{definition}[Disjunctive game]\label{def:dg}
    The rules of the disjunctive game are in \Cref{fig:dg-rules}.
Infinite runs, and runs that end in elementary disjunctive states where \I do not win according to Definition~\ref{def:win}, are winning for \You and losing for \Me.  \textbf{(Dupl)} is referred to as the \emph{duplication rule} and \textbf{(Sched)} as the \emph{scheduling}, or \emph{underlining rule}.
\end{definition}
\noindent It follows from the Gale-Stewart Theorem~\cite{gale} that every disjunctive game $\mathbf{DG}(D)$ and every \cc-disjunctive game $\mathbf{DG}^\cc(D)$ is determined.

\subsection{Adequacy and best strategies}\label{sec:adq-D}
Now we state the main result of this section, whose proof 
is split into two propositions:  The left-to-right direction in Proposition \ref{disjY} and 
 the right-to-left direction in Proposition \ref{prop:adq-D} below. 

\begin{theorem}
\label{thm:adeq} \I have a winning strategy in $\mathbf{DG}(D)$ (in $\mathbf{DG}^\cc(D)$) iff for every (\cc-) \PNL-model $\mathbb{M}$, there is some $g\in D$ such that \I have a winning strategy in $\mathbf{G}_\mathbb{M}(g)$.
\end{theorem}

\begin{corollary}
The formula $\phi$ is (\cc-) \PNL-valid iff \I have a winning strategy in $\mathbf{DG}(\mathbf{P},i:[A]\phi)$ (in $\mathbf{DG}^\cc(\mathbf{P},i:[A]\phi)$).
\end{corollary}

\begin{proposition}\label{disjY}
Let $\M$ be a (\cc-) \PNL-model. If \I have a winning strategy in $\mathbf{DG}(D)$ (in $(\mathbf{DG}^\cc(D)$), then there is some $g\in D$ such that \I have a winning strategy in $\mathbf{G}_\M(g)$.
\end{proposition}

\begin{proof}
Let $\sigma$ be \My winning strategy in $\mathbf{DG}(D)$. We show the following claim by bottom-up tree induction on $\sigma$: For every $H\in \sigma$ there is some $g\in H$ and a winning strategy $\sigma_H$ for \Me in $\mathbf{G}_\M(g)$. The case for the root, $H=D$, gives the desired result.

If $H$ is a leaf, then, since $\sigma$ is winning, there is an elementary $h\in H$ such that \I win the semantic game at $h$ over $\M$. In this case, $\sigma_H$ consists of the single leaf $h$.

If \You move in $H$, then it must be of the form $H'\bigvee \underline{h}$,
where $h$ is labeled ``Y'' in the semantic game. By  definition of strategy,
all successors of $H$ must be in $\sigma$, and they are of the form $H'\bigvee
h'$, where $h'$ ranges over all possible game states immediately after \Your
choice at $h$. By inductive hypothesis, for every $h'$,
$\sigma_{H'\bigvee h'}$ is a winning strategy for some $g\in H' \bigvee h'$. If
for some $h'$, $\sigma_{H'\bigvee h'}$ is a winning strategy for some $g\in
H'$, then we can set $\sigma_H=\sigma_{H'\bigvee h'}$. Otherwise, all
$\sigma_{H'\bigvee h'}$ is a winning strategy for $h'$. Hence, we can connect
the roots of all these trees to the new common root $h$ to obtain a winning
strategy for \Me in $\mathbf{G}_\M(h)$.

If \I move in $H$ according to the rules 
\textbf{(Dupl)} or \textbf{(Sched)}, then the resulting disjunctive state is still $H$ (except maybe some game state could be underlined or duplicated). Hence, the claim follows from the inductive hypothesis.
If \I move in $H$ according to \textbf{(Move)}, then $H=H'\bigvee
\underline{h}$, and the unique child of $H$ in $\sigma$ is $H'\bigvee h'$,
where $h'$ is a possible game state after \My move in the semantic game at $h$.
By the inductive hypothesis, $\sigma_{H'\bigvee h'}$ is a winning strategy for
\Me in $\mathbf{G}_M(g)$, for some $g\in H'\bigvee h'$. If $g\in H'$, we
proceed as above. If $g=h'$, then we can append the root of $\sigma_{H' \bigvee
h'}$ to $h$ to obtain a winning strategy for \Me in $\mathbf{G}_\M(h)$.
\end{proof}

\paragraph{My best way to play.}
We will now describe a strategy $\sigma$ for \Me for the game
$\mathbf{DG}(D_0)$. This strategy is -- in a way -- the optimal way to play the
disjunctive game. Intuitively $\sigma$ exploits all of \My possible choices
without sacrificing \My winning chances. Consequently, $\sigma$ is winning iff
\I can win the game at all. (This claim
follows from Proposition~\ref{disjY} by classical reasoning). 

Let us fix an enumeration of pairs $(g,h)$ of game states of the semantic game such that every pair appears in this enumeration infinitely often. Let us denote by $\#(g,h)$ the number of the pair $(g,h)$ under this enumeration. Throughout the game, let us keep track of the number of execution steps $n$ of $\sigma$. At $D_0$, $n=0$. The strategy $\sigma$ is as follows:

\begin{enumerate}
    \item[(C1)] If in the current disjunctive state $D$, $D^{el}$ is winning, \I end the game.
    \item[(C2)] Otherwise, let $n=\#(g,h)$. If $D=D'\bigvee g$, according to the label of  $g$ we have: 
    \\
    \noindent(a)  ``Y'' (otherwise, skip), then underline $g$ and \You make \Your move. \\
    \noindent (b)  ``I'' and $h$ is a child of $g$ in the evaluation game (otherwise, skip), then duplicate $g$, schedule a copy of $g$, and go to $h$ in that copy, i.e., the new disjunctive state is $D'\bigvee g \bigvee h$.
\item[(C3)] Increase $n$ by 1 and go to (C1).
\end{enumerate}

In words, until the game reaches a winning  state, \My
strategy is to play in a way such that \I always duplicate a state, and then play
by exhausting all possible moves in that state. 

Let $\pi$ be the run of the game
$\mathbf{DG}(D_0)$ resulting from \You playing according to \Your winning
strategy and \Me playing \My best way $\sigma$. We say that a game state $g$
\emph{appears along} $\pi$, and write $g\in \pi$, if it occurs in a disjunctive
state in $\pi$. We say that $g$ \emph{disappears}, if $g\in \pi_n$ and for some
$m> n$, $g\notin \pi_m$. The following holds (proof in Appendix~\ref{app:proofs}).

\begin{lemma}\label{lem:closure}
Let $\pi$ be as above. Then:\\
1) Let $g\in \pi$ be a non-elementary game state labelled ``Y'' in the semantic game. Then at least one successor of $g$ appears along $\pi$.\\
2) Let $g\in \pi$ be a non-elementary game state labeled ``I'' in the semantic game. Then all successors of $g$ appear along $\pi$.
\end{lemma}

We can now show that $\pi$ gives rise to a model $\M_\pi$ with the property that \You have a winning strategy for every $g$ appearing along $\pi$. 

\begin{definition}\label{emepsilon}
Let $\mathcal{E}$ be a  set of elementary game states. Let $\M_\mathcal{E}$ be the following named model:
\noindent - \textbf{Agents} $\A$: an agent $\ag_i$  for each nominal $i$ appearing in $\mathcal{E}$;

\noindent - \textbf{Accessibility relations}: $\R^-_\mathcal{E}$ is the least symmetric relation such that $\ag_i\R_\mathcal{E}^- \ag_j$ whenever (i)\textsuperscript{+} $\mathbf{O}, \_: R^- (i,j)$ is in $\mathcal{E}$. $\R^+_\mathcal{E}$ is the least reflexive symmetric relation such that (i)\textsuperscript{-} $\ag_i\R_\mathcal{E}^+ \ag_j$ whenever $\mathbf{O}, \_: R^+ (i,j)$ is in $\mathcal{E}$;

\noindent - \textbf{Valuation function} $\V_\mathcal{E}$: $ \ag_i \in \V_\mathcal{E}(p)$ iff the state $\mathbf{O}, i: p$ is in $\mathcal{E}$;

\noindent - \textbf{Assignment} $\g_\mathcal{E}$: $\g_\mathcal{E}(i)=\ag_i$.

The model $\M_\mathcal{E}^\cc$ is as $\M_\mathcal{E}$, except  $\R^+$ also has (ii)\textsuperscript{+} $\ag_i\R^+\ag_j$ if $\mathbf{P},\_:R^-(i,j)\in \mathcal{E}$ or if (iii)\textsuperscript{+} no $\mathbf{Q},\_:R^\pm(i,j)$ for $\mathbf{Q}\in \{\mathbf{P},\mathbf{O}\}$ is in $\mathcal{E}$, and is closed under reflexivity and symmetry, and $\R^-$ also has (ii)\textsuperscript{-} $\ag_iR^-\ag_j$ if $\mathbf{P},\_:R^+(i,j)\in \mathcal{E}$ and is closed under symmetry.
\end{definition}

The degree $\delta(\phi)$ of a formula $\phi$ to be 0 if $\phi$ is elementary,
$\delta(\phi_1 \star \phi_2) = \max\{\delta(\phi_1), \delta(\phi_2)\} + 1$ for
$\star$ a binary and $\delta(\triangle \psi) = \delta(\psi) + 1$ for
$\triangle$ a unary operator. We extend $\delta$ to game states by setting
$\delta(\mathbf{Q}, i : \phi ) = \delta(\phi)$ for $\mathbf{Q} \in \{\mathbf{P}, \mathbf{O}\}$.

\begin{lemma}\label{lemmamodel}
Let $\M_\pi$ be the model $\M_\mathcal{E}$ $(\M_\mathcal{E}^\cc$) from the above definition, where $\mathcal{E}$ is the set of all elementary game states appearing along $\pi$. This model is a (\cc-) \PNL~model. Furthermore, if $g$ appears along $\pi$, then \You have a winning strategy for $\mathbf{G}_{\M_\pi}(g)$.
\end{lemma}

\begin{proof}
To prove that $\M_\mathcal{E}$ is a \PNL-model it remains to show that it is non-overlapping. Towards a contradiction, assume that $(\ag_i,\ag_j)\in \R_\pi^+\cap \R_\pi^-$. By definition, both $\mathbf{O},\_:R^+(i,j)$ and $\mathbf{O},\_:R^-(i,j)$ appear\footnote{Since the relations are symmetric, we can identify $R^\pm(i,j)$ with $R^\pm(j,i)$.} along $\pi$. Since the elementary states of $\pi$
are accumulative, there is a disjunctive state $D$ in $\pi$ containing both of these states. But then \I could have won the game in $D$, since in every \PNL-model at least one of $R^+(i,j)$ and $R^-(i,j)$ must be false. This is a contradiction to the assumption that $\pi$ results from \You playing \Your winning strategy.

We show that $\M_\mathcal{E}^\cc$ is a \cc-\PNL-model. First, we note that this model is collectively connected, since, $(iii)^+$ (see \Cref{emepsilon}) is equivalent to $\neg((i)^+\vee (ii)^+\vee (i)^-\vee (ii)^-)$. Also, the model remains non-overlapping. Suppose, we had $(\ag_i,\ag_j)\in \R^+\cap \R^-$.  According to the definition of $\R^+$ and $\R^-$ all possible scenarios how this could happen lead to contradiction:

\noindent - \textbf{(1)} $(i)^+$ and $(i)^-$: At some point, both states appear in a disjunctive state in $D$. But \I would have won the game at $D$, since there is no non-overlapping model where both $R^+(i,j)$ and $R^-(i,j)$ are true.

\noindent - \textbf{(2)} $(i)^+$ and $(ii)^-$: \I win the game in a disjunctive state, where both $\mathbf{O},\_:R^+(i,j),\mathbf{P},\_:R^+(i,j)\in \mathcal{E}$ states occur.

\noindent - \textbf{(3)} $(ii)^-$ and $(i)^-$: Similar to the case above.

\noindent - \textbf{(4)} $(ii)^-$ and $(ii)^-$: \I would have won in a state containing $\mathbf{P},\_:R^+(i,j)$ and $\mathbf{P},\_:R^-(i,j)$, since there is no collectively connected models where both $R^+(i,j)$ and $R^-(i,j)$ are false.

\noindent - \textbf{(5)} $(iii)^+$, excludes $(i)^-$ and $(ii)^-$, hence $(\ag_i,\ag_j)\notin \R^-$.

We prove the second claim by induction on the degree of $g$. The elementary cases
where $g$ is of the form $\mathbf{O}, i: \phi$ follow directly from the
definition of $\M_\pi$. Assume $g= \mathbf{P}, i: p$ appears along $\pi$, but
$\M_\pi,\ag_i\Vdash p$. The latter implies that 
$\mathbf{O}, i: p$ appears along $\pi$. Reasoning as above, there is a disjunctive state $D$ in $\pi$ containing both $\mathbf{P}, i: p$ and $\mathbf{O}, i: p$, which means that $D$ would be winning for \Me. The case for $\mathbf{P}, \_: R^\pm(i,j)$ is similar.

For the inductive step, let $g\in \pi$ be non-elementary with the label ``Y''.
By Lemma~\ref{lem:closure}, some child $h$ of $g$ appears along $\pi$. By the
inductive hypothesis, there is a winning strategy $\mu_{h}$ for \You for
$\mathbf{G}_{\M_\pi}(h)$. Hence, appending the root of $\mu_h$ to $g$ gives a
winning strategy for \You in $\mu_g$ in $\mathbf{G}_{\M_\pi}(g)$.
If $g$ is non-elementary with label ``I'', then, by Lemma~\ref{lem:closure},
all children $h$ of $g$ appear along $\pi$. For each $h$ there is a winning
strategy for \You in $\mathbf{G}_{\M_\pi}(h)$. Thus, appending the roots of all
$\mu_h$ to the new common root $g$ gives a winning strategy for \You in
$\mathbf{G}_{\M_{\pi}}(g)$. 
\end{proof}

\begin{proposition}\label{prop:adq-D}
    Assume that 
 for every model $\mathbb{M}$, there is some $g\in D$ such that \I have a winning strategy in $\mathbf{G}_\mathbb{M}(g)$.
 Then, 
 \I have a winning strategy in $\mathbf{DG}(D)$.
\end{proposition}
\begin{proof}
Suppose \I do not have a winning strategy in $\mathbf{DG}(D)$.
By the Gale-Stewart Theorem, \You have a winning strategy in this game. Let \Me
play according to the strategy $\sigma$ from above, and let $\pi$ be the run through $\mathbf{DG}(D)$ resulting from \You playing \Your winning strategy and \Me playing $\sigma$. Let $\M_\pi$ be the model from
Definition~\ref{emepsilon}. By Lemma~\ref{lem:closure} and
Lemma~\ref{lemmamodel}, \You have a winning strategy for $\mathbf{G}_{\M_\pi}(g)$ for all $g\in D$. 
\end{proof}

\section{From strategies to Proofs}\label{sec:proofs}
Theorems~\ref{th:adequacy} and \ref{thm:adeq} imply that winning strategies for \Me in the disjunctive game correspond to validity. 
In this section, we will extend this result to proof systems. 
This will be done by introducing a sequent calculus $\DS$, where proofs correspond to \My winning strategies in the
disjunctive game. Before that, we will demonstrate that winning strategies
can be finitized. 

\subsection{\Your optimal choices\label{sect:opt}}
In this section, we show how to adapt the disjunctive game so that it becomes
finitely branching in ``Y''-nodes. This  will help us to conveniently
formulate the disjunctive game as a calculus. Infinite branching occurs only in
the case of the rules $\bfP_{[A]}$ and $\bfO_{\Diamond^\pm}$ where branching is
parametrized by the nominals. We will show that in these situations, there is
an optimal choice for \You, so \I can expect \You to play according to this
choice.

\begin{proposition}\label{prop:bestchoice}
    Let $j$  be a nominal different from $i$ not occurring in $D$ nor in $\phi$. Then:
\noindent \textbf{(1)}  \You have a winning strategy in $\mathbf{DG}(D \bigvee \mathbf{P},i: [A]\phi)$ iff \You have a winning strategy in $\mathbf{DG}(D \bigvee \mathbf{P},j: \phi)$, and similarly for $\mathbf{DG}^\cc$.

    \noindent \textbf{(2)} \You have a winning strategy in $\mathbf{DG}(D \bigvee \mathbf{O},i:\Diamond^\pm \phi)$ iff \You have winning strategies in both $\mathbf{DG}(D \bigvee \mathbf{P},\_: R^\pm(i,j))$ and $\mathbf{DG}(D \bigvee \mathbf{O},j:\phi)$, and similarly for $\mathbf{DG}^\cc$.
\end{proposition}

This result implies that \My winning strategies in the disjunctive game can be
finitely represented: in every disjunctive state whose children branch over the
nominals, it is enough to consider a single child only, given by a \emph{fresh} nominal $j$
not appearing in that disjunctive state. Intuitively, the proof of Proposition~\ref{prop:bestchoice} runs by transforming a winning strategy for \You in the semantic game for all game states in $D\bigvee \mathbf{P},k:\phi$ (where $k$ is arbitrary) over a model $\M$ into a winning strategy for \You for all game states in $D\bigvee \mathbf{P},j:\phi$ over a modified model $\M'$, where $j$ is according to the proposition. By the Theorems~\ref{th:adequacy} and \ref{thm:adeq}, this gives the desired result. The actual proof is carried out in the appendix. 

\subsection{The proof system $\DS$}\label{sec:seq-system}

Now we detail how to formally transform (provability) games into (sequent) proof systems.

{\em Labeled nominal formulas} are either \emph{labeled formulas} of
the form $i:\phi$ or \emph{relational atoms} of the form $R(i,j)$,
where $i$ and $j$ are nominals and $\phi$ is a \PNL~formula.\footnote{Observe that here we are abusing the notation, identifying $k:R(i,j)$ with $R(i,j)$. Recall from \Cref{foot:R} that the truth value of  these atoms depend only on the underlying model.} \emph{Labeled sequents} have the form $\Gamma \seq \Delta$, where
$\Gamma,\Delta$ are multisets containing labeled nominal formulas.

Starting with sequents, every disjunctive state of the form
\begin{center}
$
\mathbf{O},i_1: \phi_1\bigvee \ldots \bigvee \mathbf{O},i_n: \phi_n\bigvee\mathbf{P},j_1: \psi_1 \bigvee \ldots \bigvee \mathbf{P},j_m: \psi_m
$
\end{center}
 can be rewritten as the labeled sequent $\Gamma \seq \Delta$ where
$
\Gamma =\{i_1: \phi_1, \ldots, i_n: \phi_n\}\mbox{ and } \Delta =\{j_1: \psi_i,\ldots,j_m: \psi_m\}
$. 
In what follows, we will not distinguish between disjunctive states and their corresponding labeled sequent. For example,
the disjunctive game state
$\mathbf{O},i:(\dplus\dplus p \vee \dminus\dminus p)\bigvee\mathbf{P},i: \dplus p$
will be identified with the sequent  
$i:(\dplus\dplus p \vee \dminus\dminus p)\seq i: \dplus p$.

The inference rules must be tailored in such a way that {\em
proofs} in the sequent system match exactly \My\  {\em winning strategies} in
the disjunctive game. This means that the user of the proof system takes the
role of \Me, scheduling game states and choosing moves in \I-states. 
Moreover, {\em provability} in the proof system should correspond to {\em
validity} in the game. For that,  it is necessary to establish the
formal relationship between elementary game states and logical axioms.

\begin{lemma}\label{lemma:init}
    Let $\Gamma\seq \Delta$ be composed of elementary game states only. \I win the disjunctive game in $\Gamma \seq \Delta$ iff one of the following holds\footnote{\label{foot:sym}Since relations are symmetric, we will identify $R^\pm(i,j)$ with  $R^\pm(j,i)$.}

    \noindent\textbf{i.} $R^-(i,i)\in\Gamma$ or $R^+(i,i)\in\Delta$ for some $i$;

    \noindent\textbf{ii.} $\{R^+(i,j),R^-(i,j)\}\subseteq\Gamma$ for some $i\not=j$;

    \noindent\textbf{iii.} $\Gamma\cap\Delta\not=\emptyset$.

In the case of collectively connected models, additionally, 

\noindent\textbf{iv.}  $\{R^+(i,j),R^-(i,j)\}\subseteq\Delta$ for some $i\not=j$.  \end{lemma}
\begin{proof}
By definition of the disjunctive game, it is immediate that \I win the game if (iii) holds. Moreover,
\I clearly win the game if either (i) or (ii) hold, since only $R^+$ is
reflexive and since the relations are non-overlapping. Finally, if the model is collectively connected, \I clearly win the game if (iv) holds.

Suppose that (i), (ii) and (iii) do not hold. Then \You have a winning strategy in the model $\M_{\Gamma \Rightarrow \Delta}$ from Definition~\ref{emepsilon}. If additionally (iv) does not hold, we choose the model $\M_{\Gamma \Rightarrow \Delta}^\cc$.

By proceeding as in Lemma~\ref{lemmamodel}, it is easy to see that $R^+$ is reflexive, $R^\pm$ is symmetric,  both models are non-overlapping and $\M^\cc_{\Gamma \Rightarrow \Delta}$ is collectively connected.
\end{proof}

Figure~\ref{fig:calculus} presents the labeled sequent systems $\DS$ and $\DS^\cc$, with the standard initial axiom and structural/propositional  rules. The modal  rules and the relational rules $sym$ and  $ref\pm$  coincides with the modal rules originally presented by Vigan\`{o} in~\cite{Vigano:2000}, adapted to multi-relational modal logics. 

It is routine to show that the rules $no$ and $cc$ in Figure~\ref{fig:calculus} correspond to the non-overlapping and  collectively connected axioms, respectively

\qquad\qquad\qquad$
\forall i,j. \neg(R^+(i,j)\wedge R^-(i,j)) \quad\mbox{and}\quad \forall i,j. R^+(i,j)\vee R^-(i,j)
$

The following result, 
proved in Appendix~\ref{app:proofs}, entails a normal form on proofs in $\DS$, since any proof-search procedure can be restricted so to start with applications of logical rules followed by relational rules and the initial axiom.

\begin{lemma}\label{lemma:norm} 
In a bottom-up reading of derivations, the relational rules permutes up w.r.t. any other logical rule in $\DS/\DS^\cc$.
Moreover, the weakening rules below are admissible in $\DS/\DS^\cc$.

\vspace{0.15cm}

\qquad\qquad\qquad\qquad\qquad$
\infer[L_w]{\Gamma,i:\phi\seq\Delta}{\Gamma\seq\Delta} \qquad \infer[R_w]{\Gamma\seq\Delta,i:\phi}{\Gamma\seq\Delta}
$
\end{lemma}
The following result immediately implies that the disjunctive game $\mathbf{DG}$ is adequate with respect to the calculus $\DS$.

\begin{theorem}\label{thm:adequacy-sequent}
\I have a winning strategy in the disjunctive game $\mathbf{DG}(\Gamma\seq\Delta)$ (in $\mathbf{DG}^\cc(\Gamma \Rightarrow \Delta)$) iff $\Gamma\seq\Delta$ is provable in $\DS$ (in $\DS^\cc$). 
\end{theorem}
\begin{proof} The proof is by case analysis in the rules of the disjunctive game/last rule applied in the proof, and it is based in the following correspondence:

\noindent - \textbf{(Dupl)} Duplication in the game corresponds to left and right contraction (Rules $R_c$, $L_c$). 

\noindent - \textbf{(Sched)} Scheduling game moves over non-elementary formulas corresponds to choosing propositional or modal rules to be applied.
Observe that Lemma~\ref{lemma:norm} guarantees that propositional and modal rules in $\DS$ can always be chosen before dealing with the elementary case.

\noindent - \textbf{(Move)} Applying the rule chosen in \textbf{(Sched)} corresponds to applying the respective sequent rule in $\DS$.  Note that \I should be prepared to any movement from \You. Hence, branching in  \Your possible moves corresponds to branching in a sequent rule. On the other hand, infinite branching is handled as 
    explained in \Cref{sect:opt}.

\noindent - \textbf{(End)} Due to Lemma~\ref{lemma:init}, winning states are completely captured by the axioms.
\end{proof}

Let us write $\models_\PNL \Gamma \Rightarrow \Delta$ iff for every \PNL-model there is some $i:\gamma \in \Gamma$ such that $\M,\g(i)\not\vdash \gamma$, or there is some $i:\delta \in \Delta$ such that $\M,\g(i)\models \delta$. Similarly, we define $\models_{\cc\PNL} \Gamma \Rightarrow \Delta$. We have the following consequence of Theorems~\ref{th:adequacy},~\ref{thm:adeq}, and~\ref{thm:adequacy-sequent}: 

\begin{corollary}
Let $\Gamma,\Delta$ be multisets of labeled formulas. Then $\models_\PNL\Gamma \Rightarrow \Delta$ ($\models_{\cc\PNL} \Gamma \Rightarrow \Delta$) iff there is a proof of $\Gamma \Rightarrow \Delta$ in $\DS$ (in $\DS^\cc$). In particular, $\phi$ is (\cc-) \PNL-valid iff there is a proof of $\Rightarrow \phi$ in $\DS$ (in $\DS^\cc$).
\end{corollary}

\begin{figure}
    \begin{center}
\resizebox{.85\textwidth}{!}{
\begin{minipage}[t]{\textwidth}
	\headline{\sc Axiom and Structural Rules}
	
	\medskip

    \begin{center}
 \begin{prooftree}
        \hypo {}
        \infer1 [$init$]{\Gamma, i: \phi_{el} \Rightarrow  \Delta, i: \phi_{el}}
        \end{prooftree}
\qquad
 \begin{prooftree}
        \hypo {\Gamma, i:  \phi, i: \phi \Rightarrow  \Delta}
        \infer1 [$(L_c)$]{\Gamma, i: \phi \Rightarrow  \Delta}
        \end{prooftree}        
		\qquad
        \begin{prooftree}
        \hypo {\Gamma \Rightarrow  i: \phi,i: \phi,\Delta}
        \infer1 [$(R_c)$]{\Gamma \Rightarrow  i: \phi, \Delta}
        \end{prooftree}
    \end{center}
        
\headline{\sc Propositional Rules}
\begin{center}
        \begin{prooftree}
        \hypo {\Gamma \Rightarrow i: \phi, \Delta}
        \infer1 [\((L_\neg)\)]{\Gamma, i: \neg \phi \Rightarrow \Delta}
        \end{prooftree}
  \qquad
  \begin{prooftree}
        \hypo {\Gamma, i:  \phi \Rightarrow \Delta}
        \infer1[\((R_\neg)\)]{\Gamma \Rightarrow i: \neg  \phi, \Delta}
        \end{prooftree}

\bigskip

        \begin{prooftree}
        \hypo {\Gamma, i: \phi \Rightarrow \Delta}
        \hypo{\Gamma, i: \psi \Rightarrow \Delta}
        \infer2 [\((L_\vee)\)]{\Gamma, i: \phi \vee  \psi \Rightarrow \Delta}
        \end{prooftree}
        \qquad
        \begin{prooftree}
        \hypo {\Gamma \Rightarrow i:  \phi, \Delta}
        \infer1[\((R_\vee^1)\)]{\Gamma \Rightarrow i:  \phi \vee  \psi, \Delta}
        \end{prooftree}
        \qquad 
        \begin{prooftree}
        \hypo {\Gamma \Rightarrow i:  \psi, \Delta}
        \infer1[\((R_\vee^2)\)]{\Gamma \Rightarrow i:  \phi \vee  \psi, \Delta}
        \end{prooftree}
        \bigskip

        \begin{prooftree}
        \hypo {\Gamma, i: \phi \Rightarrow \Delta}
        \infer1 [\((L_\wedge^1)\)]{\Gamma, i: \phi \wedge  \psi \Rightarrow \Delta}
        \end{prooftree}
       \qquad
    	\begin{prooftree}
        \hypo {\Gamma, i: \psi \Rightarrow \Delta}
        \infer1 [\((L_\wedge^2)\)]{\Gamma, i: \phi \wedge  \psi \Rightarrow \Delta}
        \end{prooftree}
       \qquad
       \begin{prooftree}
        \hypo {\Gamma \Rightarrow i:  \phi, \Delta}
        \hypo {\Gamma \Rightarrow i:  \psi, \Delta}
        \infer2[\((R_\wedge)\)]{\Gamma \Rightarrow i:  \phi \wedge  \psi, \Delta}
        \end{prooftree}
                
    \end{center}
\headline{\sc Modal Rules}
\begin{center}
         \begin{prooftree}
        \hypo {\Gamma, R^\pm(i,j) \Rightarrow \Delta}
        \infer1 [\((L_{\Diamond^\pm})_1\)]{\Gamma, i: \Diamond^\pm \phi\Rightarrow \Delta}
        \end{prooftree}
       \qquad
       \begin{prooftree}
        \hypo {\Gamma, j:\phi \Rightarrow \Delta}  
        \infer1 [\((L_{\Diamond^\pm})_2\)]{\Gamma, i: \Diamond^\pm \phi\Rightarrow \Delta}
        \end{prooftree}
      \bigskip
      
      \begin{prooftree}
        \hypo {\Gamma \Rightarrow R^\pm(i,j), \Delta}
        \hypo {\Gamma \Rightarrow j: \phi, \Delta}
        \infer2 [\((R_{\Diamond^\pm})\)]{\Gamma\Rightarrow i: \Diamond^\pm \phi,\Delta}
        \end{prooftree}
        \quad
        \begin{prooftree}
        \hypo {\Gamma, j: \phi \Rightarrow \Delta}
        \infer1 [$(L_{[A]})$]{\Gamma,  i:[A]\phi \Rightarrow \Delta}
        \end{prooftree}
        \quad
        \begin{prooftree}
        \hypo {\Gamma \Rightarrow  j: \phi, \Delta}
        \infer1 [$(R_{[A]})$]{\Gamma\Rightarrow  i:[A]\phi, \Delta}
        \end{prooftree}

    \end{center}

\headline{\sc Relational Rules}
\begin{center}
         \begin{prooftree}
        \hypo {\Gamma \Rightarrow  \Delta, R^{\pm}(j,i)}
        \infer1 [$sym$]{\Gamma \Rightarrow  \Delta, R^{\pm}(i,j)}
        \end{prooftree}
        \qquad
\begin{prooftree}
        \hypo { }
        \infer1 [$ref+$]{\Gamma\Rightarrow  \Delta, R^{+}(i,i)}
        \end{prooftree}
        \bigskip

\begin{prooftree}
        \hypo { }
        \infer1 [$ref-$]{\Gamma, R^{-}(i,i)\Rightarrow  \Delta}
        \end{prooftree}
        \qquad
        \begin{prooftree}
        \hypo {\Gamma \Rightarrow  \Delta, R^+(i,j)}
        \hypo {\Gamma \Rightarrow  \Delta, R^-(i,j)}
        \infer2 [$no$]{\Gamma \Rightarrow  \Delta}
        \end{prooftree}
    \end{center}
        
\headline{\sc Collectively Connected (for $\DS^\cc$)}
\begin{center}
         \begin{prooftree}
        \hypo {}
        \infer1 [$cc$]{\Gamma \Rightarrow  \Delta, R^+(i,j), R^-(i,j) }
        \end{prooftree}
\end{center}
\end{minipage}
}
\end{center}
\vspace{-0.3cm}
\caption{The proof system $\DS$.
In the rule init, $\phi_{el}$ denotes an
elementary formula. In the rules   $(L_{\Diamond^\pm})_1$,
$(L_{\Diamond^\pm})_2$, and $(R_{[A]})$,  the nominal $j$ is fresh.The rule
$R_{\meddiamondminus}$ has the proviso that $i\neq j$. The system $\DS^\cc$ also includes the rule $cc$ for reasoning about collectively connected systems. \label{fig:calculus} } 
\end{figure}
The next examples shows how $\DS^\cc$ elegantly captures collectively connectedness. 
\begin{example}[Connectedness]\label{ex:corr}
    The formula $(\bplus p \to [A]p) \vee (\bminus p \to [A]p)$,
characterizing collective connectedness \cite{DBLP:journals/logcom/PedersenSA21},
has the following proof in $\DS^\cc$:\\

\qquad\noindent
\resizebox{.83\textwidth}{!}{
    $
       \infer=[R_\vee,R_w,R_\neg,L_\neg]{ \Rightarrow i:(\bplus p \to [A]p) \vee (\bminus p \to [A]p)}{
\infer=[R_{[A]}]{\Rightarrow i:[A]p, i:\dplus \neg p, i:\dminus \neg p}{
                    \infer[R_\dplus]{\Rightarrow j:p,  i:\dplus \neg p, i:\dminus \neg p}{
                        \infer[R_\dminus]{\Rightarrow j:p, i:\dminus \neg p, R^+(i,j)}{
                            \infer[cc]{\Rightarrow j:p, R^+(i,j), R^-(i,j)}{}
                            &
                            \infer=[R_\neg, init]{\Rightarrow j:p, j:\neg p, R^+(i,j)}{}
                        }
                        &
                        \infer=[R_\neg,init]{\Rightarrow j:p, j:\neg p,i:\dminus \neg p}{}
                    }
                }
}
    $
}
\end{example}
Proving cut-admissibility of labeled systems can be cumbersome due to the presence of relational rules. In~\cite{DBLP:journals/apal/MarinMPV22}, a systematic procedure for
transforming axioms into rules was presented, based on {\em focusing} and {\em
polarities}~\cite{andreoli92jlc}. This procedure not only allows for 
generalizing different approaches for transforming axioms into sequent
rules present in the literature~\cite{Sim94,Vigano:2000,Neg05}, but it also provides 
a uniform way of proving cut-admissibility for the resulting systems.

While it is out of the scope of this paper to introduce all this machinery just to prove cut-admissibility of $\DS/\DS^\cc$, we note that it is possible to directly transform the semantic description of (\cc-)\PNL~into a labeled sequent system equivalent to $\DS/\DS^\cc$, by using the methodology in~\cite{DBLP:journals/apal/MarinMPV22} and  adopting the {\em negative polarity} to atomic formulas. Hence the cut-admissibility result for $\DS/\DS^\cc$ is a particular instance of the general result in~\cite{DBLP:journals/apal/MarinMPV22}.
\begin{theorem}\label{thm:cut}
The following cut rule is admissible in $\DS/\DS^\cc$

\vspace{0.15cm}
\qquad\qquad\qquad\qquad$
\infer[cut]{\Gamma\seq\Delta}{\Gamma\seq\Delta,i:\phi & i:\phi,\Gamma\seq\Delta}
$
\end{theorem}
As a consequence, $\DS/\DS^\cc$ are consistent, since the only rules that can be applied in an empty sequent is $no$ and it is routine to show that it does not trivialize derivations. Moreover, cut-admissibility also serves as a tool for proving meta-theoretical properties: $(\bplus p \to [A]p) \vee (\bminus p \to [A]p)$ in Example~\ref{ex:corr} is not provable in $\DS$, that is, the rule $cc$ is necessary.

\section{Dynamic operators and extensions}
\label{sec:extensions}

In this section we show how the
global link-adding and local link change  modalities from
\cite{DBLP:journals/logcom/PedersenSA21} 
can be defined in our framework.
Adding such modalities requires the underlying model $\M$ to be part of the game
state, and calls for a different presentation of resulting sequent calculus.

The logic \dPNL~extends the syntax in Section \ref{sec:pnl} with the following cases:\\
\[
  \phi ::=\cdots \mid  \pnlP \phi\mid  \pnlN \phi\mid  \pnlPN \phi \mid  
 \pnlOP \phi \mid  \pnlON \phi
\]

\noindent
The global link-adding modalities, $\pnlP \phi$ (resp. $\pnlN  \phi$) is forced if, after adding a positive (resp. negative) link \emph{somewhere} in the network, $\phi$ holds. $\pnlPN$  adds either a positive or a negative link.
The local change modality 
$\pnlOP \phi$ (resp. $\pnlON \phi$) holds at agent $\ag$ whenever $\phi$ holds after changing one of the $\ag$'s
links from negative to positive (resp. from positive to negative).

In the following, if $\bbM = \langle \A,\R^+,\R^-,\V,\g\rangle$,
 we denote by with $\bbM \cup \{\ag \Rinv^+ \b\}$ the model 
$ \langle \A,\R^+ \cup \{\ag \R^+ \b, \b \R^+ \ag\},\R^-,\V,\g\rangle$.
Similarly for $\bbM \cup \{\ag \Rinv^- \b\}$. We denote by
$\bbM \setminus \{\ag \Rinv^{+} \b\}$ the model where 
both $(\ag,\b)$ and $(\b,\ag)$ are removed from $\R^+$. Similarly 
for $\bbM \setminus \{\ag \Rinv^{-}\b\}$.

The semantics of the new operators is the following \cite{DBLP:journals/logcom/PedersenSA21}:

\[
\small
\begin{array}{lll}
    \bbM,\ag \Vdash \pnlP\phi &\text{ iff } \text{  there are } \b,\c\in \A \text{ s.t. }   (\b,\c)\not\in \R^- \text{ and }  \bbM\cup\{\b \Rinv^+ \c\},\ag \Vdash  \phi\\
    \bbM,\ag \Vdash \pnlN\phi &\text{ iff } \text{  there are } \b,\c\in \A \text{ s.t. }   (\b,\c)\not\in \R^+ \text{ and }  \bbM\cup\{\b \Rinv^- \c\},\ag \Vdash  \phi\\
    \bbM,\ag \Vdash \pnlPN\phi &\text{ iff } \text{ there are } \b,\c\in \A \text{ s.t. }   (\b,\c)\not\in \R^- \text{ and }  \bbM\cup\{\b \Rinv^+ \c\},\ag \Vdash  \phi\\
                               & \text{ or there are } \b,\c\in \A \text{ s.t. }   (\b,\c)\not\in \R^+ \text{ and }  \bbM\cup\{\b \Rinv^- \c\},\ag \Vdash  \phi\\
    \bbM,\ag \Vdash \pnlOP\phi &\text{ iff } \text{ there is } \b \in \A \text{ s.t. }  (\ag,\b)\in \R^- \text{ and }  \bbM\cup\{\ag \Rinv^+ \b\} \setminus\{\ag \Rinv^- \b\},\ag \Vdash \phi\\
    \bbM,\ag \Vdash \pnlON\phi &\text{ iff } \text{ there is } \b \in \A \text{ s.t. }  \ag \neq \b, (\ag,\b)\in \R^+ \text{ and }  \bbM\cup\{\ag \Rinv^- \b\} \setminus\{\ag \Rinv^+ \b\},\ag \Vdash \phi\\
\end{array}
\]

We consider game states of the form $\mathbf{P}, \bbM, \ag:\phi$ and
$\mathbf{O}, \bbM, \ag:\phi$, where the model $\bbM$ over which the
game is played is explicit. \Cref{fig:game-dpnl} presents the rules for the
 semantic game  \dPNL. Adequacy is given by the following result, which has a similar proof to Theorem~\ref{th:adequacy}.

\begin{figure}
\begin{center}
    \resizebox{.90\textwidth}{!}
{
\begin{minipage}[t]{\textwidth}
\hrulefill
~~\begin{description}
\item[$(\bfP_{\pnlN})$] At $\mathbf{P}, \bbM,  \ag: \pnlN \phi$, 
     \I choose $\b,\c\in A$ s.t. $(\b,\c)\not\in \R^+$ and the game continues with 
    $\bfP, (\bbM\cup\{\b \Rinv^- \c\}),\ag:  \phi$. \You win if there are no such $\b$ and $\c$. 
\item[$(\bfO_{\pnlN})$] \vspace{-2mm}At $\bfO, \bbM,  \ag: \pnlN \phi$, 
     \You choose $\b,\c\in A$ s.t. $(\b,\c)\not\in \R^+$ and the game continues with 
    $\bfO, (\bbM\cup\{\b \Rinv^- \c\}),\ag:  \phi$. \I win if there are no such $\b$ and $\c$. 
\item[$(\bfP_{\pnlP})$] At $\mathbf{P}, \bbM,  \ag: \pnlP \phi$, \I choose 
     $\b,\c\in A$ s.t. $(\b,\c)\not\in \R^-$  and the game continues with 
    $\bfP, (\bbM\cup\{\b \Rinv^+ \c\}),\ag:  \phi$. \item[$(\bfO_{\pnlP})$] \vspace{-2mm}At $\bfO, \bbM,  \ag: \pnlP \phi$, 
    \You choose $\b,\c\in A$ s.t. $(\b,\c)\not\in \R^-$ and the game continues with 
    $\bfO, (\bbM\cup\{\b \Rinv^+ \c\}),\ag:  \phi$. \item[$(\bfP_{\pnlOP})$] 
    At $\bfP, \bbM,  \ag: \pnlOP \phi$, 
    \I choose $\b\in A$ s.t. $(\ag,\b)\in \R^{-}$ and the game continues with 
    $\bfP, (\bbM\cup\{\ag \Rinv^+ \b\} \setminus\{\ag \Rinv^- \b\}),\ag:  \phi$. 
    \You win if there is no such a $\b$. 
\item[$(\bfO_{\pnlOP})$] 
    \vspace{-2mm}At $\bfO, \bbM,  \ag: \pnlOP \phi$, 
     \You choose $\b\in A$ s.t. $(\ag,\b)\in \R^{-}$ and the game continues with 
    $\bfO, (\bbM\cup\{\ag \Rinv^+ \b\} \setminus\{\ag \Rinv^- \b\}),\ag: \phi$. 
    \I win if there is no such a $\b$. 
\end{description}
\vspace{-0.3cm}
\hrulefill
\end{minipage}
}
\end{center}
\vspace{-0.2cm}
\caption{Rules for \dPNL~adding/changing modalities.
The rules for $\pnlPN$, and the rules  $\bfP_{\pnlON}$ and $\bfO_{\pnlON}$ are similar and omitted.
    {The global addition modalities do not necessarily add a \emph{new} link. In
    $\bfP_{\pnlP}$, \I 
     (and \You in $\bfO_{\pnlP}$) always have a choice, say  $(\ag,\ag)\notin \R^{-}$.}
    \label{fig:game-dpnl}
}
\end{figure}

\begin{theorem}\label{th:adequacy2}
Let $\M$ be a \PNL-model, $\ag$ an agent, and $\phi$ a \dPNL~formula.
Then: (1)
 \I have a winning strategy for $\mathbf{G}_\M(\mathbf{P}, \M, \ag: \phi)$ iff $\M,\ag \models \phi$; and (2)
 \You have a winning strategy for $\mathbf{G}_{\M}(\mathbf{P}, \M, \ag: \phi)$ iff $\M,\ag\not \models \phi$.
\end{theorem}
\begin{example}\label{ex:consensus}
    Let $\M_2$ be as in \Cref{ex:balance}. 
    \I have a winning strategy for the game state $\bfP,\bbM_2,\ag : \pnlON (4B)$:
    \I just need to change the relation  $\ag \R^+ \c$ to obtain 
    the model $\M_1$ (where \I have a winning strategy for $4B$). 
     \You do not have a winning strategy for $\bfO,\bbM_2,\ag : \pnlON \pnlON (4B)$
    (and \I win $\bfP,\bbM_2,\ag : \neg(\pnlON \pnlON (4B))$). In words, 
    \You cannot enforce \emph{balance} by making $\ag$ disagree with her two friends. 
    Finally, the formula $[A]\bminus \bot$ characterizes ``reconciliation'' in a network \cite{DBLP:journals/logcom/PedersenSA21}, 
    where there are no disagreements between agents. \I have a winning strategy in the game 
    $\bfP, \bbM_2,\ag : \dplus \pnlOP [A] \bminus\bot$. (See the outputs of the tool in  \Cref{ap:examples} and \cite{tool}). 
\end{example}

Defining a sequent system for \dPNL\ would require passing through a
provability game $\mathbf{dDG}$ as done in Section~\ref{sec:dis-game}. We 
skip this step since it is similar to the case of \PNL. Instead, we will go
directly to the design of a proof system, which turns out to be a non-trivial
task.

The problem is that the new modalities update the relational values and multisets contexts are not adequate for handling  this {\em linear} behavior.  Hence relational atoms will be stored in a separate {\em linear context}, where
information can be updated. 
A {\em relational context} $\mathcal{R}$ is a \emph{set} containing only relational predicates, and a {\em relational sequent} has the form
$
\rs{\rc}{\Gamma}{\Delta}
$, where $\Gamma,\Delta$ are multisets of labeled formulas (hence no relational atoms).

 The label sequent system $\dDS$~for
\dPNL~is depicted in \Cref{fig:dDS}. The rules for the other connectives
are obtained by adapting those in  \Cref{fig:calculus} with relational contexts.  Rule $L_{\Diamond^\pm}$ introduces the  predicate
$R(i,j)$, for a fresh $j$, into the  context $\mathcal{R}$. The
proviso of the rules for the 
global adding-link modalities guarantee non-overlapping, and 
 the rules
for $\pnlOP$ forbid adding into $\mathcal{R}$ the atom $R^-(i,i)$. Moreover, 
the only way of adding new elements
into $\mathcal{R}$ is using the rules $L_{\Diamond^\pm}$ and those for $\pnlP$,
$\pnlN$ and $\pnlPN$. This explains the additional hypothesis in the theorem
below, which is proved in Appendix~\ref{app:proofs}.

\begin{figure}
\noindent
\hspace{-0.2cm}
\resizebox{.82\textwidth}{!}{
\noindent\begin{minipage}[t]{\textwidth}
    \begin{tabular}{ccc}
	{
         \begin{prooftree}
        \hypo {\rs{\rc,R^+(i,j)}{\Gamma,k:\phi}{\Delta}}
        \infer1 [$(L_{\pnlP})$]{\rs{\rc}{\Gamma,k:\pnlP \phi}{\Delta}}
        \end{prooftree}
        }
        & 
         {\begin{prooftree}
        \hypo {\rs{\rc,R^+(i,j)}{\Gamma}{\Delta,k:\phi}}
        \infer1 [$(R_{\pnlP})$]{\rs{\rc}{\Gamma}{\Delta,k:\pnlP \phi}}
        \end{prooftree}
        }
&
        {
        \begin{prooftree}
        \hypo {\rs{\rc,R^-(i,j)}{\Gamma,k:\phi}{\Delta}}
        \infer1 [$(L_{\pnlN})$]{\rs{\rc}{\Gamma,k:\pnlN k:\phi}{\Delta}}
        \end{prooftree}
        }
   \\\\
        {  \begin{prooftree}
        \hypo {\rs{\rc,R^-(i,j)}{\Gamma}{\Delta,k:\phi}}
        \infer1 [$(R_{\pnlN})$]{\rs{\rc}{\Gamma}{\Delta,k:\pnlN \phi}}
        \end{prooftree}
        }
&
	{
         \begin{prooftree}
        \hypo {\rs{\rc,R^+(i,j)}{\Gamma,k:\phi}{\Delta}}
        \infer1 [$(L_{\pnlOP})$]{\rs{\rc,R^-(i,j)}{\Gamma,k:\pnlOP \phi}{\Delta}}
        \end{prooftree}
        }
        & 
        {\begin{prooftree}
        \hypo {\rs{\rc,R^+(i,j)}{\Gamma}{\Delta,k:\phi}}
        \infer1 [$(R_{\pnlOP})$]{\rs{\rc,R^-(i,j)}{\Gamma}{\Delta,k:\pnlOP \phi}}
        \end{prooftree}
        }
\\\\        
        {
        \begin{prooftree}
        \hypo {\rs{\rc,R^-(i,j)}{\Gamma,k:\phi}{\Delta}}
        \infer1 [$(L_{\pnlON})$]{\rs{\rc,R^+(i,j)}{\Gamma,k:\pnlON \phi}{\Delta}}
        \end{prooftree}
        \medskip
        }
        &
        {\begin{prooftree}
        \hypo {\rs{\rc,R^-(i,j)}{\Gamma}{\Delta,k:\phi}}
        \infer1 [$(R_{\pnlON})$]{\rs{\rc,R^+(i,j)}{\Gamma}{\Delta,k:\pnlON \phi}}
        \end{prooftree}
        }
        &
        {
         \begin{prooftree}
         \hypo {\rs{\rc, R^\pm(i,j) }{\Gamma, j:\phi}{\Delta}}
             \infer1 [\((L_{\Diamond^\pm})\)]{\rs{\rc}{\Gamma, i: \Diamond^\pm \phi}{ \Delta}}
        \end{prooftree}
	}
\end{tabular}
\end{minipage}
}
\vspace{-0.2cm}
\caption{System $\dDS$. Rules $L/R_{\pnlP}$  (resp. $L/R_{\pnlN}$) have the proviso that $R^-(i,j)\not\in\mathcal{R}$ (resp.
$R^+(i,j)\not\in\mathcal{R}$), modulo symmetry (see \Cref{foot:sym}). In rules $L/R_{\pnlON}$, $i\neq j$. 
Rules for $\pnlPN$ are similar and omitted. In $L_{\Diamond^\pm}$, $j$ is fresh. 
\label{fig:dDS}}
\end{figure}

\begin{theorem}\label{thm:adequacy-sequent-dyn}
Let  $\Gamma$ and $\Delta$ be multisets of \dPNL~formulas
not containing  relational predicates.
\I have a winning strategy in the disjunctive game $\mathbf{dDG}(\Gamma\seq\Delta)$ iff $\Gamma\seq\Delta$ is provable in $\dDS$.
\end{theorem}
 
\section{Concluding Remarks}\label{sec:conc}
We have introduced two new techniques for \PNL~\cite{DBLP:journals/logcom/PedersenSA21}, with the aim of
formally reasoning about positive and negative relations among agents
and group polarization: a satisfiability game that allows for the
verification of properties within concrete networks of agents; and
a validity game with the corresponding cut-free sequent calculus. 
Our
contributions offer promising avenues for automated reasoning, as demonstrated
by our prototypical tool \cite{tool}. Furthermore, by showing that reasoning about frame
properties of the underlying model can be
delayed until reaching elementary games/formulas, we can modularly adapt to
different relational properties. 

Currently, we are exploring extensions that relax
symmetry assumptions, allowing for representing situations where agent $a$ may
influence the opinion of $b$ but not the other way around. Additionally, we are
investigating the concept of ``budget'' as in the game proposed in
\cite{DBLP:conf/tableaux/LangOPF19} to characterize scenarios where proponents
and opponents operate within a limited \emph{political capital}, where
adding/changing relations can potentially decrease such a capital. 
To this end, the preference of spending as little capital as possible could be expressed in a combination of $\PNL$ with a suitable \emph{choice logic}, i.e., a logic where preferences are definable at the object level. Semantic games for choice logics have been investigated in \cite{Freiman2023TruthLogic} and the lifting of game-induced choice logic, \textbf{GCL}, to a provability game and proof system was demonstrated in \cite{Freiman2023}.
Finally, following
the techniques developed in \cite{DBLP:journals/jlap/OlartePR23} 
for analyzing sequent systems in rewrite logic, we are extending 
our tool \cite{tool} to also support the sequent calculi proposed here. 

This work can be seen as a continuation of a program of lifting semantic games
to analytic calculi \cite{DBLP:journals/sLogica/FermullerM09,Pavlova2021}. Our
approach is a refinement of previous work on modal logic
\cite{DBLP:conf/wollic/Freiman21,HybrJour}  as it replaces model checking at
the level of axioms with explicit rules for the classes of $\PNL$ and
$\cc$-$\PNL$ models. We therefore  provide hand-tailored systems for reasoning about
group polarization and opens up the aforementioned routes to mechanization.

 \bibliographystyle{plain}

\appendix

\section{Some selected proofs}\label{app:proofs}

{\bf Theorem~\ref{th:adequacy}.}
\textit{Let $\M$ be a \PNL-model, $\ag$ an agent, and $\phi$ a formula.
\begin{enumerate}
\item \I have a winning strategy for $\mathbf{G}_\M(\mathbf{P}, \ag: \phi)$ iff $\M,\ag \models \phi$. 
\item \You have a winning strategy for $\mathbf{G}_{\M}(\mathbf{P}, \ag: \phi)$ iff $\M,\ag\not \models \phi$.
\end{enumerate}
}
\begin{proof}
Both directions of (1) and (2) are shown simultaneously by induction on the
degree\footnote{Observe that, at this point of the text, the degree has not been
introduced yet. Its definition can be found just above Lemma~\ref{lemmamodel}
on page 8.} of the game state $g=\mathbf{P}, \ag:\phi$. 

If $\phi$ is elementary, then the result trivially follows by the definition. 

If $g=\mathbf{P}, \ag:\phi_1\wedge\phi_2$ then 
\I have a winning strategy for  $\mathbf{G}_\M(\mathbf{P}, \ag:\phi_1\wedge\phi_2)$ iff \I have 
winning strategies for both $\mathbf{G}_\M(\mathbf{P}, \ag:\phi_1)$ and $\mathbf{G}_\M(\mathbf{P}, \ag:\phi_2)$. By the inductive
hypothesis, this is the case iff $\M,\ag \models \phi_1$ and $\M,\ag \models \phi_2$, which is equivalent to $\M,\ag \models \phi_1\wedge\phi_2$.

If $g=\mathbf{P}, \ag: \Diamond^\pm \phi$, \I have a winning strategy for $\mathbf{G}_\M(\mathbf{P}, \ag: \Diamond^\pm \phi)$ iff there is a $\R^\pm$-successor $\b$ of $\ag$ and
\I have a winning strategy for  $\mathbf{P},\b:\phi$. By the inductive hypothesis, this occurs iff  $\M,\b \models \phi$, but then 
$\M,\ag \models \Diamond^\pm \phi$. 

The other cases are similar.
\end{proof}

\bigskip

\noindent
{\bf Lemma~\ref{lem:closure}.}
\textit{Let $\pi$ be as above. Then:\\
1) Let $g\in \pi$ be a non-elementary game state labelled ``Y'' in the semantic game. Then at least one successor of $g$ appears along $\pi$.\\
2) Let $g\in \pi$ be a non-elementary game state labeled ``I'' in the semantic game. Then all successors of $g$ appear along $\pi$.
}

\begin{proof}
First, note that since \You play according to \Your winning strategy,
$\pi$ does not end in a winning disjunctive state whose elementary party is winning for \Me. This means that case (C1) in the definition of $\sigma$ is never reached.

\begin{enumerate}
    \item Suppose $g$ appeared in $\pi$ at stage $n\geq 0$ in the above construction. Since every pair appears in the enumeration infinitely often, there is some minimal $m \geq n$ such that $m=\#(g,h)$, for some $h$. At step $m$ in the execution of $\sigma$ against \Your winning strategy, the current disjunctive state is of the form $D'\bigvee g$. According to $\sigma$, \I underline $g$ and \You move to some successor $h'$, according to \Your winning strategy. This means the new game state is of the form $D' \bigvee h'$, hence $h'$ is the successor of $g$ appearing along $\mathfrak{h}$.
    \item Suppose $g$ appeared in $\pi$ at stage $n\geq 0$. Now we additionally fix an arbitrary successor $h$ of $g$ in the evaluation game. By the properties of $\#$, there is a minimal $m\geq n$ such that $m=\#(g,h)$. Since \I always first duplicate game states labeled ``I'', before \I make a move into them, $g$ does not disappear. Hence, at step $m$ in the execution of $\sigma$, the current disjunctive state is of the form $D'\bigvee g$. According to $\sigma$, \I duplicate $g$ and go to $h$ in one copy, i.e. the new disjunctive state is $D'\bigvee g \bigvee h$, which shows that $h$ appears along $\pi$. 
\end{enumerate} 
\end{proof}

\bigskip

\noindent {\bf Proposition \ref{prop:bestchoice}} { \it Let $j$  be a nominal different from $i$ not occurring in $D$ nor in $\phi$. Then:
\noindent \textbf{(1)}  \You have a winning strategy in $\mathbf{DG}(D \bigvee \mathbf{P},i: [A]\phi)$ iff \You have a winning strategy in $\mathbf{DG}(D \bigvee \mathbf{P},j: \phi)$, and similarly for $\mathbf{DG}^\cc$.

    \noindent \textbf{(2)} \You have a winning strategy in $\mathbf{DG}(D \bigvee \mathbf{O},i:\Diamond^\pm \phi)$ iff \You have winning strategies in both $\mathbf{DG}(D \bigvee \mathbf{P},\_: R^\pm(i,j))$ and $\mathbf{DG}(D \bigvee \mathbf{O},j:\phi)$, and similarly for $\mathbf{DG}^\cc$.
    }

The proof is split into a sequence of lemmas. First, we need to define substitutions formally. For a sequence $x$ of finite or infinite length, let $x_n$ denote its $n$-th element (if defined),  and let $\mathrm{range}(x)=\{x_i:i\in \mathbb{N}\}$. Let $\phi$ be a formula and $a$ and $b$ two sequences of nominals of the same length, where every nominal occurs only once in each sequence. We define 
$\phi[a/b]$ as the formula obtained by simultaneously substituting for every number $n$ all occurrences of $a_n$ in $\phi$ with $b_n$. For example, let $a = \langle i, j \rangle$, $b = \langle k, l \rangle$ and $\phi = R^+(i,j) \vee R^-(j,l)$. Then $\phi[a/b] = R^+(k,l)\vee R^-(l,l)$.  As another example let $a = \langle i_1, i_2, ... \rangle$ and $b = \langle i_2, i_3, ... \rangle$. Then $R^+(i_1,i_2)[a/b] = R^+(i_2,i_3)$, since the substitution happens simultaneously. We extend the notion of substitution to game states: for a game state $g=\mathbf{Q}, i: \phi$ of the evaluation game and two sequences of nominals $a,b$, we define the substitution $g[a/b]$ as $\mathbf{Q}, i[a/b]:\phi[a/b]$. Similarly to histories, strategies, and disjunctive states.

\begin{lemma}\label{lem:techicalnom}
Let $\mathbb{M}_1=(\A,\R^+,\R^-,\V,\g_1)$ and $\mathbb{M}_2=(\A,\R^+,\R^-,\V,\g_2)$ be named and $\g_2(i[b/a])=\g_1(i)$ for all nominals $i$. Then for all game states $g$, $\mathbf{G}_{\mathbf{M}_1}(g)\cong \mathbf{G}_{\mathbb{M}_2}(g[b/a])$.
\end{lemma}
\begin{proof}
By the assumption, $\g_2$ is surjective, even if restricted to
$N[b/a]=\{i[b/a]:i\in N\}$. By Proposition~\ref{prop:surjg}, it is, therefore,
enough to prove $\mathbf{G}_{\mathbb{M}_1}(g)\cong
\mathbf{G}^{N[b/a]}_{\mathbb{M}_2}(g[b/a])$. We proceed by  
induction on the degree of $g$.

If $g$ is elementary and of the form $\mathbf{P}, i: j$ then it is winning for \Me over $\mathbb{M}_1$ if and only if $\g_1(i) = \g_1(j)$. By assumption, this is equivalent to $\g_2(i[b/a])=g_2(i[b/a])$, which means that $\mathbf{O},i[b/a]:j[b/a]$ is winning for me over $\mathbb{M}_2$. The other elementary cases are similar.

As an example of a simple induction step, we consider $\mathbf{P},i:\phi_1 \vee
\phi_2$. If \I have a winning strategy for this game state over
$\mathbb{M}_1$, then there is some $k\in \{1,2\}$ such that \I have a winning
strategy in $\mathbf{P},i:\phi_k$. By the inductive hypothesis, \I have a
winning strategy in $\mathbf{P},i[b/a]:\phi_k[b/a]$ over $\mathbb{M}_2$. Hence,
\I have a winning strategy in $(\mathbf{P},i:\phi_1 \vee \phi_2)[b/a])$ over
that model. The other direction is similar.

The most interesting induction step is for the modal rules, so let us consider
$\mathbf{O},i:\dplus \psi$. Suppose, \I have a winning strategy in
$\mathbf{G}_{\mathbb{M}_1}(\mathbf{O},i:\dplus \psi)$. Then, for every nominal
$j$, \I have winning strategies in
$\mathbf{G}_{\mathbb{M}_1}(\mathbf{P},j:R(i,j))$ and
$\mathbf{G}_{\mathbb{M}_1}(\mathbf{O},j: \psi)$. By the inductive hypothesis, \I
have winning strategies in
$\mathbf{G}^{N[b/a]}_{\mathbb{M}_2}(\mathbf{P},j[b/a]:R(i[b/a],j[b/a]))$ and
$\mathbf{G}^{N[b/a]}_{\mathbb{M}_2}(\mathbf{O},j[b/a]: \psi[b/a])$. In other
words, \I have winning strategies in
$\mathbf{G}^{N[b/a]}_{\mathbb{M}_2}(\mathbf{P},k:R(i[b/a],k)$ and
$\mathbf{G}^{N[b/a]}_{\mathbb{M}_2}(\psi[b/a])$, for every $k\in N[b/a]$. Since
branching in this game is restricted over $N[b/a]$, we conclude that \I have a
winning strategy in
$\mathbf{G}^{N[b/a]}_{\mathbb{M}_2}((\mathbf{P},i:\dplus \psi)[b/a])$. The other
direction, as well as the other cases of induction steps, are similar.
\end{proof}

\begin{lemma}\label{lem:samename}
If $\g(k) = \g(l)$, then $\mathbf{G}_\mathbb{M}(g)\cong\mathbf{G}_\mathbb{M}(g[k/l])$.
\end{lemma}

\begin{proof}
We show that $\g(i[k/l])=\g(i)$ for all nominals $i$. If $i\ne k$, then $\g(i[k/l])=\g(i)$. If $i=k$, then by the assumption, $\g(i[k/l])=\g(l)=\g(k)=\g(i)$. The statement of the lemma follows from this fact and Lemma~\ref{lem:techicalnom}.
\end{proof}

For a model $\mathbb{M}$ and two sequences of nominals $a,b$, let $\mathbb{M}[a/b]$ be the same as $\mathbb{M}$, except for the denotation function: $\g_{[a/b]}(i) = \g(i[a/b])$.

\begin{lemma}\label{lem:substsurj}
Let $\mathbb{M}$ be named and $a,b$ two sequences of nominals with $\mathrm{range}(a)\subseteq\mathrm{range}(b)$. Then $\mathbb{M}{[a/b]}$ is $N[b/a]$-named. Furthermore, $\mathbf{G}_\mathbb{M}(g)\cong \mathbf{G}_{\mathbb{M}[a/b]}(g[b/a])$.
\end{lemma}

\begin{proof}
We have to show that $\g_{[a/b]}$ is surjective when restricted to $N[b/a]=\{i[b/a]:i\in N\}$. Let $\ag$ be an agent and $i$ its name under $\g$. If $i\notin \mathrm{range}(b)$, then $i\notin \mathrm{range}(b)$ and we have
$i[b/a][a/b]=i[a/b]=i$. 
If $i \in \mathrm{b}$, then $i=b_m$ for some $m$. Then
$i[b/a][a/b]=b_m[b/a][a/b]=a_m[a/b]=b_m=i$. 
This shows that $\g_{[a/b]}(i[b/a])=\g(i[b/a][a/b])=\g(i)$, i.e. $\ag$ has a
name in $N[b/a]$ under $g_{[a/b]}$. This identity together with
Lemma~\ref{lem:techicalnom} also implies the strategic equivalence of the game
$\mathbf{G}_\mathbb{M}(g)$ and $\mathbf{G}_{\mathbb{M}[a/b]}(g[b/a])$.
\end{proof}

We are now ready to prove Proposition~\ref{prop:bestchoice}.

\begin{proof}[Proposition~\ref{prop:bestchoice}]
We will show (2), since (1) is similar and simpler. From right to left is trivial. For the other direction, assume that  \You have a winning strategy in $\mathbf{DG}(D \bigvee \mathbf{O},i:\dplus \phi)$ with $j$ as in the assumption. By Theorem~\ref{thm:adeq}, there is a named model $\mathbb{M}$ such that \You have winning strategies in $\mathbf{G}_\mathbb{M}(g)$ for all $g\in  D$ and in $\mathbf{O}, i: \dplus \phi$. The latter implies that \You have winning strategies in $\mathbf{G}_\mathbb{M}(\mathbf{O}, \_:  R^\pm(i,k))$ and $\mathbf{G}_\mathbb{M}(\mathbf{O}, k: \phi)$ for some nominal $k$. 

Let $j_1, j_2, ...$ be a sequence of nominals not occurring in $D$ or $\phi$ and different from $k$, $j$, and $i$. Let $a = \langle j,j_1, j_2,...\rangle$ and $b = \langle k,j,j_1,j_2,...\rangle$. We have that $\mathrm{range}(a)\subseteq \mathrm{range}(b)$, therefore Lemma~\ref{lem:substsurj} applies.  We have the following chain of equivalences:

\begin{align*}
\mathbf{G}_\mathbb{M}(\mathbf{O},k: \phi)    &\cong \mathbf{G}_{\mathbb{M}[a/b]}(\mathbf{O},k[b/a]:\phi[b/a]) &&\text{by Lemma~\ref{lem:substsurj}}\\
    &= \mathbf{G}_{\mathbb{M}[a/b]}(\mathbf{O},j:\phi[k/j]) &&\text{by conditions on } i,j,k\\
    &= \mathbf{G}_{\mathbb{M}[a/b]}(\mathbf{O},j[k/j]:\phi[k/j])\\
    &\cong \mathbf{G}_{\mathbb{M}[a/b]}(\mathbf{O},j: \phi) &&\text{by Lemma~\ref{lem:samename}} \text{ and }\g_{[a/b]}(j)=\g_{[a/b]}(k)
\end{align*}

A similar argument shows that $\mathbf{G}_\mathbb{M}(\mathbf{O},\_:
R^\pm(i,k))\cong \mathbf{G}_{\mathbb{M}[a/b]}(\mathbf{O},\_: R^\pm(i,j))$. By
this equivalence and the assumption, \You have winning strategies in
$\mathbf{O},\_:R^\pm(i,j)$ and $\mathbf{O},j:\phi$ over $\mathbb{M}[a/b]$.
Moreover, we obtain a winning strategy for \You for $g\in D$ by using the equivalence
 $
\mathbf{G}_\mathbb{M}(g)\cong 
\mathbf{G}_{\mathbb{M}[a/b]}(g[b/a]) \cong 
\mathbf{G}_{\mathbb{M}[a/b]}(g[k/j])\cong 
\mathbf{G}_{\mathbb{M}[a/b]}(g)
$, 
the same lemmas as before and the fact that no nominals from $a$ appear in $g$.
Since \You have winning strategies for the semantic games for every $g\in D$
and $\mathbf{O},\_:R^\pm(i,j)$ over $\M[a/b]$, \You have a winning strategy in
$\mathbf{DG}(D\bigvee \mathbf{O},\_:R^\pm(i,j))$, by Proposition~\ref{thm:adeq}
and the determinacy of the game. Similarly, we conclude that \You  have a
winning strategy in $\mathbf{DG}(D\bigvee \mathbf{O},j:\phi)$. \end{proof}

\bigskip

\noindent
{\bf Lemma~\ref{lemma:init}.}
\textit{ Let $\Gamma\seq \Delta$ be composed of elementary game states only. \I win the disjunctive game in $\Gamma \seq \Delta$ iff one of the following holds
\begin{itemize}
\item[i.] $R^-(i,i)\in\Gamma$ or $R^+(i,i)\in\Delta$ for some $i$;
\item[ii.] $\{R^+(i,j),R^-(i,j)\}\subseteq\Gamma$ for some $i\not=j$;
\item[iii.] $\Gamma\cap\Delta\not=\emptyset$.
\end{itemize}
In the case of collectively connected models, additionally, 
\begin{itemize}
\item[iv.]  $\{R^+(i,j),R^-(i,j)\}\subseteq\Delta$ for some $i\not=j$  \end{itemize}}

\begin{proof}
$\M_{\Gamma \Rightarrow \Delta}$ is non-overlapping: this follows by (ii). $\M_{\Gamma\Rightarrow \Delta}^\cc$ is collectively connected, since  $(iii)^+$ in the definition of $\R^+$ equivalent to $\neg((i)^+\vee (ii)^+\vee (i)^-\vee (ii)^-)$. $\M^\cc_{\Gamma \Rightarrow \Delta}$ is non-overlapping:  Suppose, $(\ag_i,\ag_j)\in \R^+\cap \R^-$. This is impossible since all possible cases in the definitions, in which $\ag_i$ and $\ag_j$ are connected by both relations, are excluded by our assumptions:
 
\begin{itemize}
    \item   $(i)^+$ and $(i)^-$: Excluded by $\neg(ii)$.
    \item $(i)^+$ and $(ii)^-$: Excluded by $\neg (iii)$.
    \item $(ii)^-$ and $(i)^-$: Excluded by $\neg (iii)$.
    \item $(ii)^-$ and $(ii)^-$: Excluded by $\neg (iv)$.
    \item $(iii)^+$, excludes $(i)^-$ and $(ii)^-$, hence $(\ag_i,\ag_j)\notin \R^-$.
\end{itemize}

By definition of the models, $\ag_i\in \V(p)$, whenever $i:p\in \Gamma$, $\ag_i\notin \V(p)$, whenever $i:p\in \Delta$, $(\ag_i, \ag_j)\in \R^\pm$, whenever $R^\pm(i,j)\in \Gamma$. If $R^\pm(i,j)\in \Delta$, then $(\ag_i,\ag_j)\in \R^\mp$. Since both models are non-overlapping, $(\ag_i,\ag_j)\notin \R^\pm$. Hence, all game states in $\Gamma \Rightarrow \Delta$ are winning for \You.
\end{proof}

\bigskip

\noindent
{\bf Lemma~\ref{lemma:norm}.}
\textit{  
The following weakening rules are admissible in $\DS$
\[
\infer[L_w]{\Gamma,i:\phi\seq\Delta}{\Gamma\seq\Delta} \qquad \infer[R_w]{\Gamma\seq\Delta,i:\phi}{\Gamma\seq\Delta}
\]
Moreover, in a bottom-up reading of derivations, the relational rules permutes up w.r.t. any other logical rule in $\DS$.
}

\begin{proof} The proof of weakening is standard. The
proof of permutability is by straightforward case analysis. For example, 
\[
\infer[no]{\Gamma, i: \dminus \phi\Rightarrow \Delta}
{\infer[(L_{\dminus})_2]{\Gamma, i: \dminus \phi\Rightarrow \Delta,R^+(k,l)}{\deduce{\Gamma, j: \phi\Rightarrow \Delta,R^+(k,l)}{\pi_1}}&\deduce{\Gamma, i: \dminus \phi\Rightarrow \Delta,R^-(k,l)}{\pi_2}}
\]
with $j$ free, can be transformed to
\[
\infer[L_c]{\Gamma, i: \dminus  \phi\Rightarrow \Delta}
{\infer[(L_{\dminus})_2]{\Gamma, i: \dminus\phi, i: \dminus \phi\Rightarrow \Delta}
{\infer[no]{\Gamma, j:\phi, i: \dminus\phi \Rightarrow \Delta}
{\deduce{\Gamma, j:\phi, i: \dminus\phi \Rightarrow  \Delta, R^+(k,l)}{\pi_1^w}&\deduce{\Gamma, j:\phi, i: \dminus\phi \Rightarrow  \Delta, R^-(k,l)}{\pi_2^w}}}}\]
where $\pi_1^w, \pi_2^w$ are the weakened versions of $\pi_1,\pi_2$ respectively.
\end{proof}

\noindent{\bf Theorem~\ref{thm:adequacy-sequent-dyn}.}
\textit{Let  $\Gamma$ and $\Delta$ be multisets of \dPNL~formulas
not containing  relational predicates.
\I have a winning strategy in the disjunctive game $\mathbf{dDG}(\Gamma\seq\Delta)$ iff $\Gamma\seq\Delta$ is provable in $\dDS$.}

\begin{proof}
First of all, we observe that the rule 
\begin{center}
 \begin{prooftree}
         \hypo {R^\pm(i,j),\Gamma, j:\phi\seq \Delta}
             \infer1 [\((L'_{\Diamond^\pm})\)]{\Gamma, i: \Diamond^\pm \phi\seq \Delta}
        \end{prooftree}
\end{center}        
is admissible in $\DS$. In fact, this is an easy consequence of the presence of the contraction rules. Hence, although the rule $(L_{\Diamond^\pm})$ in $\dDS$ is not directly defined via a provability game, it is equivalent to its contracted version.
The rest of the proof follows the same lines as the proof in Theorem~\ref{thm:adequacy-sequent}.
\end{proof}

 \section{Examples}\label{ap:examples}

Below we present  \My strategy for winning the game on  model $\M_1$ in \Cref{ex:balance}.
\footnote{The notation $[\checkmark,\texttt{Q}]$ means that \I have a winning strategy starting in that state
where $Q\in \{\I, \You\}$ moves. 
The notation $[\times,\texttt{Q}]$ means that \You have a winning strategy
where $Q\in \{\I, \You\}$ moves. 
}

\begin{Verbatim}[fontsize=\scriptsize]
python main.py examples/model-M1.maude a "lb(p)"

[✔,Y] : P @ a : (¬ (◆ ◆ p ∨ ◇ ◇ p) ∨ ◆ p) ∧ ¬ (◆ ◇ p ∨ ◇ ◆ p) ∨ ◇ p
├── [✔,I] : P @ a : ¬ (◆ ◆ p ∨ ◇ ◇ p) ∨ ◆ p
│   └── [✔,I] : P @ a : ◆ p
│       └── [✔,I] : P @ b : p
└── [✔,I] : P @ a : ¬ (◆ ◇ p ∨ ◇ ◆ p) ∨ ◇ p
    └── [✔,I] : P @ a : ¬ (◆ ◇ p ∨ ◇ ◆ p)
        └── [✔,Y] : O @ a : ◆ ◇ p ∨ ◇ ◆ p
            ├── [✔,Y] : O @ a : ◆ ◇ p
            │   ├── [✔,Y] : O @ a : ◇ p
            │   │   └── [✔,Y] : O @ c : p
            │   └── [✔,Y] : O @ b : ◇ p
            │       └── [✔,Y] : O @ c : p
            └── [✔,Y] : O @ a : ◇ ◆ p
                └── [✔,Y] : O @ c : ◆ p
                    └── [✔,Y] : O @ c : p
\end{Verbatim}

\I do not have a winning strategy for the game on the model 
$\M_2$ in \Cref{ex:balance}: 
\begin{Verbatim}[fontsize=\scriptsize]
python main.py examples/model-M2.maude a "lb(p)" --tree

[❌,Y] : P @ a : (¬ (◆ ◆ p ∨ ◇ ◇ p) ∨ ◆ p) ∧ ¬ (◆ ◇ p ∨ ◇ ◆ p) ∨ ◇ p
├── [❌,I] : P @ a : ¬ (◆ ◇ p ∨ ◇ ◆ p) ∨ ◇ p
│   ├── [❌,I] : P @ a : ¬ (◆ ◇ p ∨ ◇ ◆ p)
│   │   └── [❌,Y] : O @ a : ◆ ◇ p ∨ ◇ ◆ p
│   │       ├── [❌,Y] : O @ a : ◆ ◇ p
│   │       │   ├── [❌,Y] : O @ c : ◇ p
│   │       │   │   └── [❌,Y] : O @ b : p
│   │       │   ├── [✔,Y] : O @ a : ◇ p
│   │       │   └── [✔,Y] : O @ b : ◇ p
│   │       │       └── [✔,Y] : O @ c : p
│   │       └── [✔,Y] : O @ a : ◇ ◆ p
│   └── [❌,I] : P @ a : ◇ p
└── [✔,I] : P @ a : ¬ (◆ ◆ p ∨ ◇ ◇ p) ∨ ◆ p
    ├── [❌,I] : P @ a : ¬ (◆ ◆ p ∨ ◇ ◇ p)
    │   └── [❌,Y] : O @ a : ◆ ◆ p ∨ ◇ ◇ p
    │       ├── [❌,Y] : O @ a : ◆ ◆ p
    │       │   ├── [❌,Y] : O @ a : ◆ p
    │       │   │   ├── [❌,Y] : O @ b : p
    │       │   │   ├── [✔,Y] : O @ a : p
    │       │   │   └── [✔,Y] : O @ c : p
    │       │   ├── [❌,Y] : O @ b : ◆ p
    │       │   │   ├── [❌,Y] : O @ b : p
    │       │   │   └── [✔,Y] : O @ a : p
    │       │   └── [✔,Y] : O @ c : ◆ p
    │       │       ├── [✔,Y] : O @ a : p
    │       │       └── [✔,Y] : O @ c : p
    │       └── [✔,Y] : O @ a : ◇ ◇ p
    └── [✔,I] : P @ a : ◆ p
        ├── [❌,I] : P @ a : p
        ├── [❌,I] : P @ c : p
        └── [✔,I] : P @ b : p
\end{Verbatim}

And \I certainly win in the negated formula: 

\begin{Verbatim}[fontsize=\scriptsize]
python main.py examples/model-M2.maude a "~ lb(p)"
[✔,I] : P @ a : ¬ ((¬ (◆ ◆ p ∨ ◇ ◇ p) ∨ ◆ p) ∧ ¬ (◆ ◇ p ∨ ◇ ◆ p) ∨ ◇ p)
└── [✔,I] : O @ a : (¬ (◆ ◆ p ∨ ◇ ◇ p) ∨ ◆ p) ∧ ¬ (◆ ◇ p ∨ ◇ ◆ p) ∨ ◇ p
    └── [✔,Y] : O @ a : ¬ (◆ ◇ p ∨ ◇ ◆ p) ∨ ◇ p
        ├── [✔,I] : O @ a : ¬ (◆ ◇ p ∨ ◇ ◆ p)
        │   └── [✔,I] : P @ a : ◆ ◇ p ∨ ◇ ◆ p
        │       └── [✔,I] : P @ a : ◆ ◇ p
        │           └── [✔,I] : P @ c : ◇ p
        │               └── [✔,I] : P @ b : p
        └── [✔,Y] : O @ a : ◇ p
\end{Verbatim}

The winning strategy for the game 
$\bfP,\bbM_2,\ag : \pnlON (4B)$
is the following: 

\begin{Verbatim}[fontsize=\scriptsize]
python main.py examples/model-M2.maude a " (-) lb(p)"
[✔,I] : P @ a : (-)((¬ (◆ ◆ p ∨ ◇ ◇ p) ∨ ◆ p) ∧ ¬ (◆ ◇ p ∨ ◇ ◆ p) ∨ ◇ p)
└── [✔,Y] : P @ a : (¬ (◆ ◆ p ∨ ◇ ◇ p) ∨ ◆ p) ∧ ¬ (◆ ◇ p ∨ ◇ ◆ p) ∨ ◇ p
    ├── [✔,I] : P @ a : ¬ (◆ ◆ p ∨ ◇ ◇ p) ∨ ◆ p
    │   └── [✔,I] : P @ a : ¬ (◆ ◆ p ∨ ◇ ◇ p)
    │       └── [✔,Y] : O @ a : ◆ ◆ p ∨ ◇ ◇ p
    │           ├── [✔,Y] : O @ a : ◆ ◆ p
    │           │   ├── [✔,Y] : O @ a : ◆ p
    │           │   │   ├── [✔,Y] : O @ a : p
    │           │   │   └── [✔,Y] : O @ c : p
    │           │   └── [✔,Y] : O @ c : ◆ p
    │           │       ├── [✔,Y] : O @ a : p
    │           │       └── [✔,Y] : O @ c : p
    │           └── [✔,Y] : O @ a : ◇ ◇ p
    │               └── [✔,Y] : O @ b : ◇ p
    │                   ├── [✔,Y] : O @ a : p
    │                   └── [✔,Y] : O @ c : p
    └── [✔,I] : P @ a : ¬ (◆ ◇ p ∨ ◇ ◆ p) ∨ ◇ p
        └── [✔,I] : P @ a : ◇ p
            └── [✔,I] : P @ b : p
\end{Verbatim}

and \I can win the following game in \Cref{ex:consensus}.

\begin{Verbatim}[fontsize=\scriptsize]
python main.py examples/model-M2.maude a " ~ ( (-) (-) lb(p))"
[✔,I] : P @ a : ¬ ((-)(-)((¬ (◆ ◆ p ∨ ◇ ◇ p) ∨ ◆ p) ∧ ¬ (◆ ◇ p ∨ ◇ ◆ p) ∨ ◇ p))
└── [✔,Y] : O @ a : (-)(-)((¬ (◆ ◆ p ∨ ◇ ◇ p) ∨ ◆ p) ∧ ¬ (◆ ◇ p ∨ ◇ ◆ p) ∨ ◇ p)
    ├── [✔,Y] : O @ a : (-)((¬ (◆ ◆ p ∨ ◇ ◇ p) ∨ ◆ p) ∧ ¬ (◆ ◇ p ∨ ◇ ◆ p) ∨ ◇ p)
    │   └── [✔,I] : O @ a : (¬ (◆ ◆ p ∨ ◇ ◇ p) ∨ ◆ p) ∧ ¬ (◆ ◇ p ∨ ◇ ◆ p) ∨ ◇ p
    │       └── [✔,Y] : O @ a : ¬ (◆ ◆ p ∨ ◇ ◇ p) ∨ ◆ p
    │           ├── [✔,I] : O @ a : ¬ (◆ ◆ p ∨ ◇ ◇ p)
    │           │   └── [✔,I] : P @ a : ◆ ◆ p ∨ ◇ ◇ p
    │           │       └── [✔,I] : P @ a : ◇ ◇ p
    │           │           └── [✔,I] : P @ c : ◇ p
    │           │               └── [✔,I] : P @ b : p
    │           └── [✔,Y] : O @ a : ◆ p
    │               └── [✔,Y] : O @ a : p
    └── [✔,Y] : O @ a : (-)((¬ (◆ ◆ p ∨ ◇ ◇ p) ∨ ◆ p) ∧ ¬ (◆ ◇ p ∨ ◇ ◆ p) ∨ ◇ p)
        └── [✔,I] : O @ a : (¬ (◆ ◆ p ∨ ◇ ◇ p) ∨ ◆ p) ∧ ¬ (◆ ◇ p ∨ ◇ ◆ p) ∨ ◇ p
            └── [✔,Y] : O @ a : ¬ (◆ ◆ p ∨ ◇ ◇ p) ∨ ◆ p
                ├── [✔,I] : O @ a : ¬ (◆ ◆ p ∨ ◇ ◇ p)
                │   └── [✔,I] : P @ a : ◆ ◆ p ∨ ◇ ◇ p
                │       └── [✔,I] : P @ a : ◇ ◇ p
                │           └── [✔,I] : P @ c : ◇ p
                │               └── [✔,I] : P @ b : p
                └── [✔,Y] : O @ a : ◆ p
                    └── [✔,Y] : O @ a : p
\end{Verbatim}

In the same example, this is \My winning strategy for 
    $\bfP, \bbM_2,\ag : \dplus \pnlOP [A] \bminus\bot$:

\begin{Verbatim}[fontsize=\scriptsize]
python main.py examples/model-M2.maude a "<+> (+) [A] [-] (p /\ ~ p)"
[✔,I] : P @ a : ◆ (+)[A]¬ (◇ ¬ (p ∧ ¬ p))
└── [✔,I] : P @ b : (+)[A]¬ (◇ ¬ (p ∧ ¬ p))
    └── [✔,Y] : P @ b : [A]¬ (◇ ¬ (p ∧ ¬ p))
        ├── [✔,I] : P @ a : ¬ (◇ ¬ (p ∧ ¬ p))
        │   └── [✔,Y] : O @ a : ◇ ¬ (p ∧ ¬ p)
        ├── [✔,I] : P @ b : ¬ (◇ ¬ (p ∧ ¬ p))
        │   └── [✔,Y] : O @ b : ◇ ¬ (p ∧ ¬ p)
        └── [✔,I] : P @ c : ¬ (◇ ¬ (p ∧ ¬ p))
            └── [✔,Y] : O @ c : ◇ ¬ (p ∧ ¬ p)
\end{Verbatim}

\end{document}